\documentclass[preprint,10pt]{elsarticle}
\usepackage[utf8]{inputenc}
\usepackage[usenames,dvipsnames]{color}
\usepackage{subfigure,wrapfig}
\usepackage{url}
\usepackage{graphicx}
\usepackage{amsmath,amssymb}
\usepackage{amsthm}
\usepackage{bm}
\usepackage{algorithm}
\usepackage{tikz}
\usetikzlibrary{arrows}
\usepackage{xcolor}
\usepackage{scrextend}

\DeclareGraphicsExtensions{.pdf,.jpg,.png}

\usepackage{thmtools,thm-restate}
\declaretheoremstyle[bodyfont=\normalfont,postheadspace=1em,qed=\qedsymbol]{examplestyle}
\declaretheorem{theorem}
\declaretheorem[within=section]{lemma}

\declaretheorem[sibling=lemma,name=Claim]{claim}

\declaretheorem[sibling=theorem,name=Question]{question}
\declaretheorem[sibling=lemma,name=Definition]{definition}
\declaretheorem[sibling=lemma,name=Example,style=examplestyle]{ex}

\newcommand{\figref}[1]{Figure \ref{fig:#1}}
\newcommand{\algref}[1]{Algorithm \ref{alg:#1}}

\newcommand{\lemref}[1]{Lemma \ref{lemma:#1}}

\newcommand{\theoref}[1]{Theorem \ref{theo:#1}}
\newcommand{\secref}[1]{Section \ref{sec:#1}}
\renewcommand{\eqref}[1]{Equation ({\ref{eq:#1}})}

\newcommand{\lemlab}[1]{\label{lemma:#1}}

\newcommand{\theolab}[1]{\label{theo:#1}}
\newcommand{\seclab}[1]{\label{sec:#1}}
\newcommand{\eqlab}[1]{\label{eq:#1}}

\newcommand{\PP}{\mathbb{P}}
\newcommand{\RR}{\mathbb{R}}

\newcommand{\bp}{{\bf p}}
\newcommand{\bx}{{\bf x}}

\newcommand{\disjointUnion}{\sqcup}

\usepackage{ifpdf}
\def\final{1}
\ifpdf
\ifnum\final=0
\newcommand{\JScomment}[1]{\marginpar{\pdfannot{/Type /Annot /Subtype /Text /T (Jessica ) /C [ 1 0 1 ] /Contents (#1) }}}
\else
\newcommand{\JScomment}[1]{}
\fi
\ifpdf
\ifnum\final=0
\newcommand{\ASJcomment}[1]{\marginpar{\pdfannot{/Type /Annot /Subtype /Text /T (Audrey  ) /C [ 0 1 0 ] /Contents (#1) }}}
\else
\newcommand{\ASJcomment}[1]{}
\fi
\ifpdf
\ifnum\final=0
\newcommand{\LTcomment}[1]{\marginpar{\pdfannot{/Type /Annot /Subtype /Text /T (Louis  ) /C [ 0 1 1 ] /Contents (#1) }}}
\else
\newcommand{\LTcomment}[1]{}
\fi

\usepackage{listings}
\usepackage{color}

\definecolor{dkgreen}{rgb}{0,0.6,0}
\definecolor{gray}{rgb}{0.5,0.5,0.5}
\definecolor{mauve}{rgb}{0.58,0,0.82}

\journal{}

\begin{document}

\begin{frontmatter}

\title{Algorithms for detecting dependencies and rigid subsystems for CAD\tnoteref{adg}} \tnotetext[adg]{An extended abstract of this article appeared as ``Detecting dependencies in geometric constraint systems'' in the 10th International Workshop on Automated Deduction in Geometry (ADG 2014), Coimbra, Portugal, 2014.}

\author[JFadd]{James Farre\fnref{JFfund}}
\ead{jamesrfarre@gmail.com}
\author[MHC]{Helena Kleinschmidt}
\ead{klein23h@mtholyoke.edu}
\author[MHC]{Jessica Sidman\fnref{JSfund}}
\ead{jsidman@mtholyoke.edu}
\author[MHC]{Audrey St. John\fnref{ASJfund}}
\ead{astjohn@mtholyoke.edu}
\author[MHC]{Stephanie Stark\fnref{SSfund}}
\ead{stark23s@mtholyoke.edu}
\author[LTadd]{Louis Theran\fnref{LTfund}}
\ead{louis.theran@aalto.fi}
\author[MHC]{Xilin Yu\fnref{XYfund}}
\ead{yu25x@mtholyoke.edu}

\address[JFadd]{Department of Mathematics, University of Utah, Salt Lake City, USA}
\address[MHC]{Mount Holyoke College, South Hadley, MA USA}
\address[LTadd]{Aalto Science Institute and Department of Computer Science \\ Aalto University, 00076 Aalto, Finland}
\fntext[JFfund]{Partially supported by NSF DMS-0849637.}
\fntext[JSfund]{Partially supported by NSF DMS-0849637 and the Hutchcroft fund.}
\fntext[ASJfund]{Partially supported by the Clare Boothe Luce Foundation, NSF DMS-0849637, NSF IIS-1253146 and the Hutchcroft fund.}
\fntext[SSfund]{Partially suppored by the Mount Holyoke College Lynk Program.}
\fntext[LTfund]{Partially supported by the European Research Council under the European Union’s Seventh Framework Programme (FP7/2007-2013) / ERC grant agreement no 247029-SDModels, Academy of Finland (AKA) project COALESCE, and the Hutchcroft fund.}
\fntext[XYfund]{Partially supported by NSF IIS-1253146, an Aalto Science Institute (AScI) Internship, and the Mount Holyoke College Lynk Program.}

\begin{abstract}
Geometric constraint systems underly popular Computer Aided Design software. Automated
approaches for detecting dependencies in a design are critical for developing robust
solvers and providing informative user feedback, and we provide algorithms for two types of
dependencies. First, we give a pebble game algorithm for detecting generic dependencies.
Then, we focus on identifying the ``special positions" of a design in which
generically independent constraints	become dependent.
We present combinatorial algorithms
for identifying subgraphs associated to factors of a particular polynomial, whose vanishing
indicates a special position and resulting dependency.
Further factoring in the Grassmann-Cayley algebra may allow a geometric
interpretation giving conditions (e.g., ``these two lines being parallel cause
a dependency'') determining the special position.
\end{abstract}


\end{frontmatter}

\section{Introduction}
Constraint-based Computer Aided Design (CAD) software, such as the popular SolidWorks program, allows
engineers to create designs using intuitive geometric constraints. When a user adds a constraint
that is \emph{dependent}, the resulting system is \emph{over-constrained}. To provide
useful feedback, efficient approaches are required to detect the minimal sub-system containing
the dependency. In this paper, we present graph-based algorithms for decomposing the underlying combinatorial
structure of a system of CAD constraints. This decomposition allows us to associate polynomials whose vanishing indicates the existence of dependencies to subsets of constraints.

\subsection{Motivation}
Automated methods for the detection and resolution of dependencies in a CAD system are
important for the underlying solver as well as the user.
Adding a constraint to a fully defined sub-system
with no relative motion among its parts, or \emph{rigid block},
results in a dependency.
The rigidity models of
2D \emph{bar-and-joint} and $d$-dimensional \emph{body-and-bar} are well-known for having
combinatorial characterizations of \emph{generically rigid} frameworks (a \emph{bar} imposes a
distance constraint between a pair of points).
However, a combinatorial characterization for 3D bar-and-joint generic rigidity remains a
conspicuously open problem.
One property highlighted as a potential barrier is the existence of \emph{contextually rigid components} \cite{chengSitharamStreinu09}, %
which do not appear in 2D bar-and-joint or $d$-dimensional body-and-bar frameworks.
A rigid component is a vertex-maximal rigid block, and a
\emph{contextually rigid component} is one that is not rigid as an induced framework.
\figref{triplebanana} depicts the well-known 3D bar-and-joint ``triple banana'' example, in which
each ``banana'' (a
$K_5$ without an edge) is a rigid component; there is a fourth contextually rigid
component formed by $\{A,B,C\}$.
Contextually rigid components may arise in CAD systems \emph{even in the plane}
where generic rigidity is understood combinatorially in all dimensions \cite{stjohnSidman}.
We give an example of this phenomenon in \figref{2DcontextuallyRigidComp} which depicts a system of
3 rigid bodies in the plane; bodies $B$ and $C$ form a rigid component, but this subsystem is flexible
as an induced framework.

\begin{figure}
\subfigure[The well-known ``triple banana'' 3D bar-and-joint framework spanning 12 joints is flexible as each ``banana'' can rotate. This framework has a contextually rigid block \{$A$, $B$, $C$\} that is flexible as an induced framework, since there are no constraints among $A,B$ and $C$.]{\begin{minipage}[b]{.45\linewidth}\centering\includegraphics[width=.5\linewidth]
{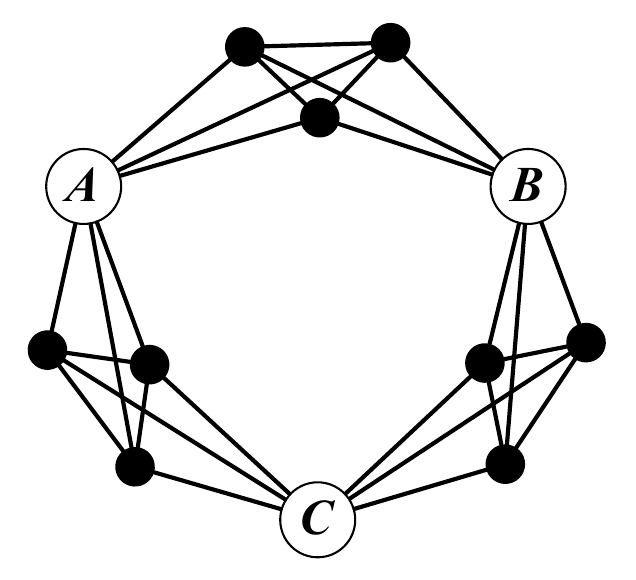}\end{minipage}\label{fig:triplebanana}}
\hfil
\subfigure[A flexible 2D body-and-cad framework consisting of 3 bodies with the following 4 constraints: dashed lines on $A$ and $B$ must be parallel; solid lines on $A$ and $C$ must be parallel;
two bars between bodies $B$ and $C$ fix the distance between pairs of points. The contextually rigid block \{B,C\} is flexible as an induced framework.]{\begin{minipage}[b]{.45\linewidth}\centering\includegraphics[width=.6\linewidth]
{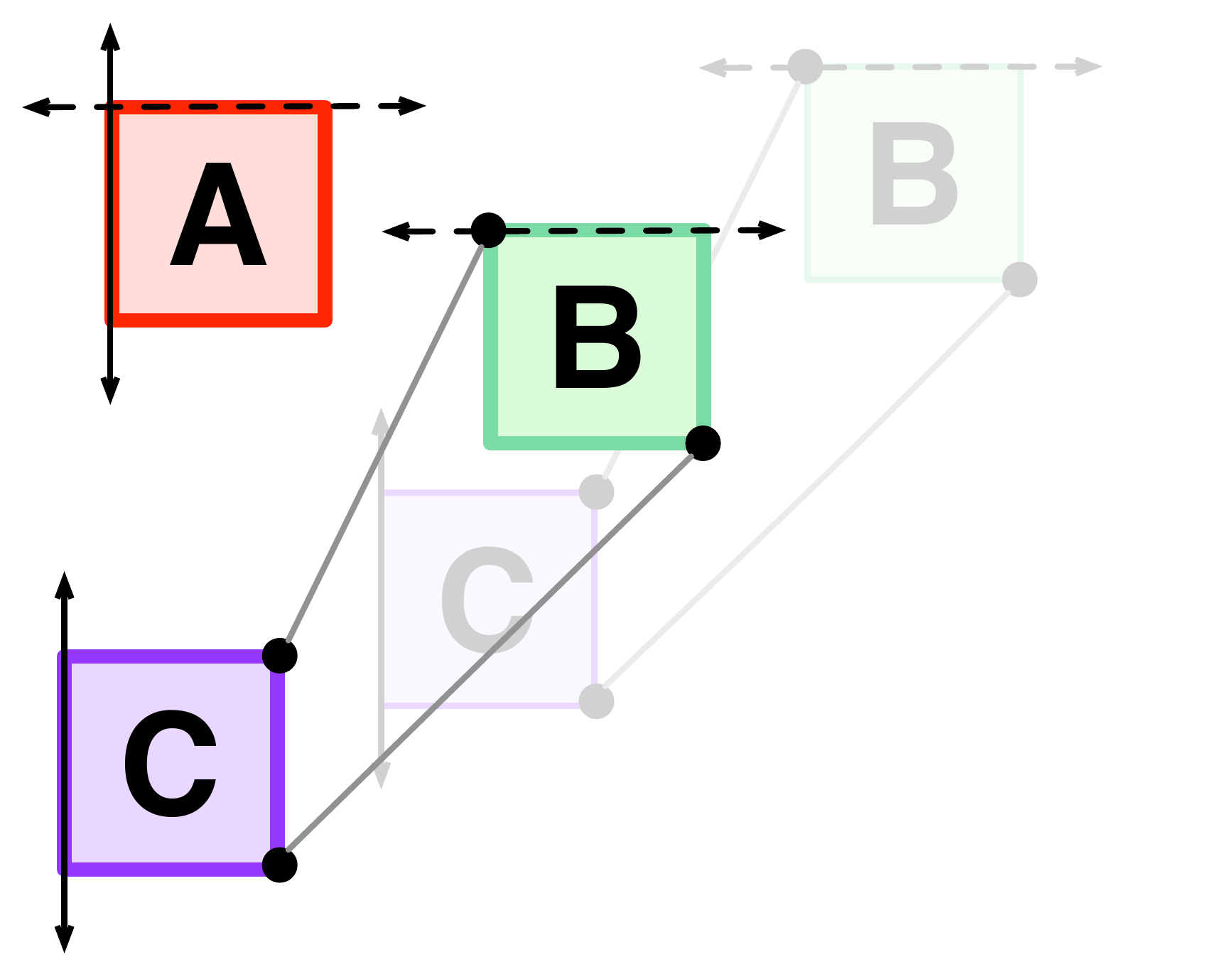}\end{minipage}\label{fig:2DcontextuallyRigidComp}}
\caption{Contextually rigid blocks highlight behavior that does not appear in 2D bar-and-joint and $d$-dimensional body-and-bar rigidity models.}
\label{fig:contextuallyRigid}
\end{figure}

In \figref{gen3Body2D} we change the parallel constraints for the framework in
\figref{2DcontextuallyRigidComp} to line-line coincidence constraints to produce a rigid framework that demonstrates
a related problem in detecting dependencies for a system that is in a \emph{special position}.
SolidWorks correctly identifies the system in \figref{gen3Body2D} as ``Fully Defined," and the framework satisfies the \emph{genericity} assumptions of the associated \emph{body-and-cad}
rigidity model; adding any constraint will result in a dependency.
However, simply changing the attachment point of one
bar results in a special (\emph{non-generic}) position that is \emph{flexible},
and SolidWorks correctly identifies it as ``Under Defined''; see \figref{swSketchSpecial}.
The constraint is now dependent, but
its consistency with the rest of the system permits a motion.
Embedding the same special position in the 3D Assembly environment highlights the difficulty that SolidWorks
appears to have; even though the design should permit the same motion, it is now detected as ``Over Defined''
(see \figref{swAssemblySpecial}).
The underlying combinatorics of the rigid and flexible
systems are the same, and it is the \emph{geometry of the constraints (the two bars are parallel) that
leads to the special position and associated dependency}.

\begin{figure}[tb]
\subfigure[SolidWorks' 2D Sketch environment identifies the {\bf rigid} \emph{generic} embedding as ``Fully Defined.'']{\begin{minipage}[b]{.48\linewidth}\centering
\includegraphics[scale=.45]{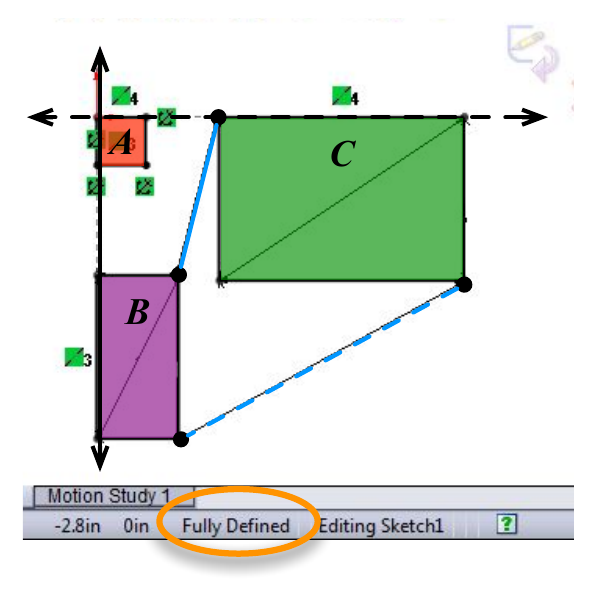}\end{minipage}\label{fig:swSketchGenRigid}}
\hfil
\subfigure[The underlying combinatorics of the framework is captured by a multigraph with a vertex for each body and labeled edge for each constraint.]{\begin{minipage}[b]{.48\linewidth}\centering
\includegraphics[scale=.45]{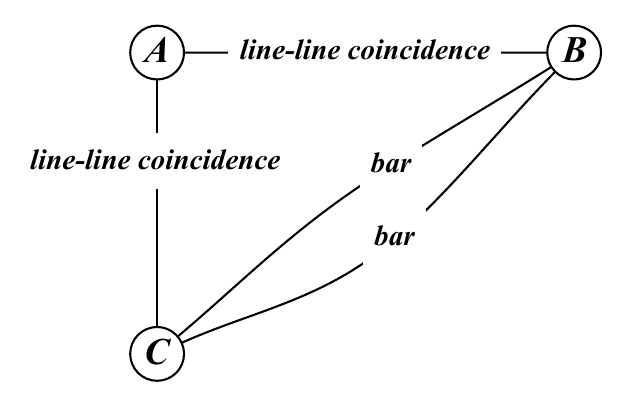}\end{minipage}\label{fig:threeBodyComb}}
\caption{A generically rigid 2D body-and-cad framework consisting of 3 bodies with the following 4 constraints: a line-line coincidence along the horizontal dashed line for $A$ and $B$; a line-line coincidence along the vertical solid line for $A$ and $C$; two bars between bodies $B$ and $C$.}
\label{fig:gen3Body2D}
\end{figure}
\begin{figure}[tb]
\subfigure[SolidWorks' 2D Sketch environment gives an identification of ``Under Defined,'' but does not reliably allow the user to explore the allowable motion. The faded position was found by suppressing the dependent constraint, then investigating the motion before unsuppressing it.]{\begin{minipage}[b]{.48\linewidth}\centering
\includegraphics[scale=.45]{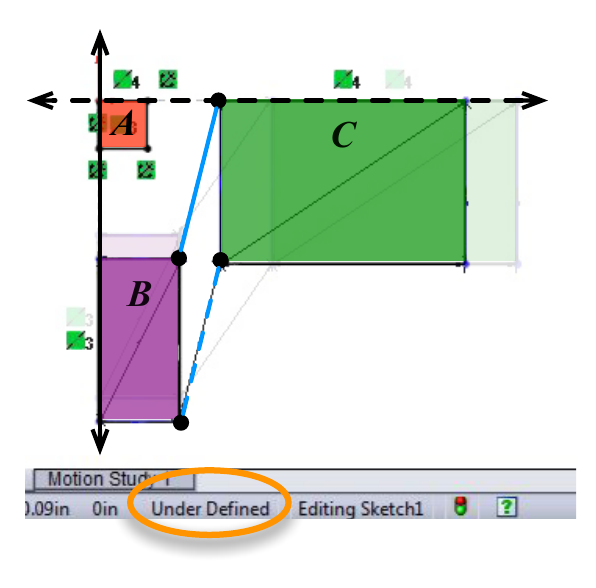}\end{minipage}\label{fig:swSketchSpecial}}
\hfil
\subfigure[SolidWorks' 3D Assembly environment gives a different identification of ``Over Defined'' for the same design embedded in 3D. No motion can be explored without suppressing the dependent constraint.]{\begin{minipage}[b]{.48\linewidth}\centering
\includegraphics[scale=.45]{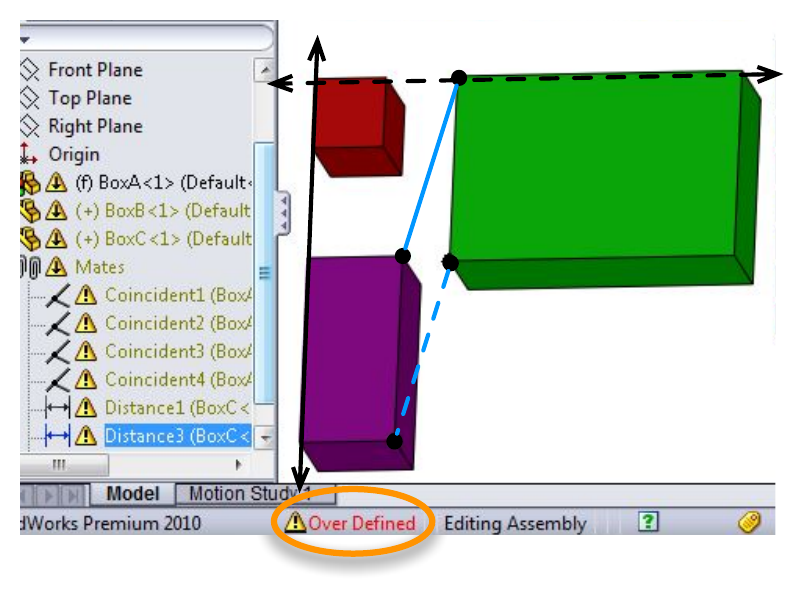}\end{minipage}\label{fig:swAssemblySpecial}}
\caption{Commercial CAD software, such as SolidWorks, is unreliable when presented with a  \emph{flexible special position}. This system shares the underlying combinatorics of \figref{gen3Body2D}, but the placement of one bar (dashed) is changed to be dependent on the other constraints. The (consistent) dependency arises due to the \emph{geometric condition} that the two bars between $B$ and $C$ are parallel.}
\label{fig:special3Body2D}
\end{figure}

\subsection{Related work}

In bar-and-joint and body-and-bar rigidity theory, constraints are specified by fixing the distance between pairs of points
and can be represented by quadratic equations.
Many other geometric constraints can also be represented by polynomial equations,
and it is possible to use algebraic methods to study systems of geometric constraints and their consequences.  See \cite{chou} for an introduction
to automated theorem proving based on Wu's method.  More recent work of \cite{geissSchreyer} uses ideas from algebraic geometry to find
special positions of the Stewart-Gough platform, and generalized Stewart-Gough platforms (with angle constraints allowed) were studied in \cite{gaoLeiLiaoZhang}.

However, working algebraically with polynomials is limited because it is computationally intensive.  Hence, much work in
rigidity theory instead focuses on \emph{infinitesmial rigidity}.  A structure is \emph{generically minimally infinitesimally} rigid if a certain
\emph{rigidity matrix} has maximal rank for some realization of the underlying graph.
As a matrix drops rank on a closed set, \emph{almost every} framework with the same combinatorics is infinitesimally
rigid (and therefore rigid).  Thus, one may detect generic dependencies numerically, by picking
random realizations.  This is the approach taken by the ``witness method'' of \cite{michelucciFoufou}.  The drawback
of numerical methods is that fast, stable, algorithms, such as SVD, do not identify the support of a
minimal dependency while those based on Gaussian elimination are not stable.
(This can
be overcome by using finite fields and the Schwartz Lemma, as discussed in \cite{GHT10}.)

Generic rigidity of a system is defined
via the rank of the rigidity matrix, which we analyze by studying a polynomial called the \emph{pure condition}.
The pure condition of a body-and-bar framework was introduced by White and Whiteley
\cite{whiteWhiteley}, who showed how to interpret the irreducible factors combinatorially and how
to describe some special positions using synthetic geometry via the Grassmann-Cayley algebra.

For certain structural models
more robust combinatorial characterizations of generic rigidity are known.
Particularly
relevant here are results for geometric constraints arising in CAD:
2D point-line frameworks \cite{owenJacksonPTLine}
and body-and-cad frameworks in 2D and 3D (omitting point-point coincidences) \cite{stjohnSidman}.
Combinatorial counting conditions arise as necessary conditions for rigidity theory,
usually in terms of a family of
``sparse matroidal graphs'' that are fundamental in generic rigidity theory (see
\cite[Appendix]{whiteley:Matroids:1996}, which reports work of White and Whiteley).
Associated pebble game algorithms can be used to check rigidity and detect components, relying on the matroidal property of \cite{leeStreinu,leeStreinuTheranCCCG}.
For body-and-bar \cite{tay} %
and 2D bar-and-joint \cite{L70} frameworks, the counting
conditions are generically sufficient as well.
\ASJcomment{Tried to cite Owen and Jackson in this paragraph}
Owen and Jackson show how to adapt Edmonds' matroid union
algorithm to the framework of pebble games for 2D point-line frameworks in \cite{owenJacksonPTLine}.
A combinatorial characterization of 3D bar-and-joint rigidity remains an open problem;
while the network flow approach of \cite{sitharamZhou2004} gives a polynomial time algorithm for the related
concept of \emph{module-rigidity},
the class of module-rigid graphs does not include rigid
\emph{nucleation-free graphs} \cite{chengSitharamStreinu09}.

In this paper, we
will see that by combining combinatorial and algebraic viewpoints, we can extract even more information
and analyze \emph{special positions} which are not covered by the combinatorial theorems.

\subsection{Contributions}
We present algorithms for detecting dependencies in CAD systems modeled as
body-and-cad frameworks. The first is a pebble game
algorithm that can check for \emph{generic dependencies} via the
combinatorial property from \cite{stjohnSidman};
when a constraint is determined to be dependent, we additionally detect its
\emph{fundamental circuit} (minimal set of constraints involved in the dependency).
To adapt the pebble game to our setting, the algorithm needs to
partition the edges in a graph and maintain $(a,a)$-sparsity on one part
and $(b,b)$-sparsity on the other; this may require dynamically adjusting the
partitions.  To prove that our algorithm is correct, we show
that it is implementing Knuth's matroid partitioning algorithm \cite{Knuth:1973:MP:891978}.

Additionally, a framework that is generically minimally
rigid (thus containing no generic dependencies) may be in a
flexible \emph{special position} caused by a non-generic dependency.
Since a special position is indicated by the vanishing of a framework's pure condition,
we develop algorithms for finding graph minors, which we refer to as \emph{factor graphs},
corresponding to its factors.  In the body-and-bar setting of \cite{whiteWhiteley}, irreducible
factor graphs correspond to irreducible isostatic subframeworks.  However, in our setting we
may have irreducible factors that do not have a natural interpretation as the pure condition
of any subframework; in particular, contextually rigid blocks give rise to these.  \JScomment{added ``sub"}
\ASJcomment{added comment about contextually rigid blocks.. I think this is true?}

\subsection{Overview}
We begin with an introduction to body-and-cad rigidity theory in \secref{background}. In \secref{PC},
we give two combinatorial descriptions of the pure condition.  We explore
the correspondence between factors of the pure condition and graph minors in \secref{factors}.
\secref{algs} contains the algorithms for detecting {\bf generic dependencies} and
factor graphs, whose vanishing determines when a framework is in a {\bf special position}. We conclude
with a case study in \secref{caseStudy}, where we provide geometric conditions for
special positions, and
conclusions and open questions in \secref{OpenQuestions}.
\section{Background}\seclab{background}
In this section, we review the fundamentals of body-and-cad rigidity theory; for full technical details,
refer to \cite{halleretal} and \cite{stjohnSidman}.

\subsection{Body-and-cad frameworks}
A {\em body-and-cad framework} consists of $n$ full-dimensional bodies with pairwise \emph{c}oincidence,
\emph{a}ngular, or \emph{d}istance constraints between them; these {\em cad} constraints are specified
between \emph{geometric elements} which are affine
linear spaces (e.g., a point, line or plane in 3D) rigidly affixed
to the bodies.  The allowed motions of a body-and-cad framework are continuous motions of the bodies
that preserve the given constraints.  A body-and-cad framework is \emph{rigid} when all of the allowed
motions are trivial, i.e., they consist of applying the same rigid-body motion to each of the bodies;
otherwise it is \emph{flexible}.

\subsubsection{CAD constraints and primitive constraints.}
There are 9 different constraint types in $2$D and $21$ in $3$D.
Examples of constraints in $2$D are: {\bf point-point distance} (a bar), {\bf point-line coincidence},
{\bf point-line distance}, {\bf line-line coincidence}, and {\bf line-line angular}.
Each constraint represents \emph{one or
more equations} restricting the relative motion of the bodies involved.  A constraint can then
be further decomposed into \emph{primitive constraints}, which correspond to single equations.
Let $d$ be the dimension of the ambient space.  Primitive constraints come in two types, which require distinct algebraic treatment: (i)
\emph{blind} constraints, that can potentially restrict any of the $\binom{d+1}{2}$ relative degrees of freedom; (ii)
\emph{angular} constraints, that restrict only the $\binom{d}{2}$ relative \emph{rotational} degrees
of freedom.

\subsection{Examples}
To give some intuition about body-and-cad rigidity, consider the planar body-and-cad framework in Figure \ref{fig:simple2Dexample}\footnote{Note that results in this paper apply in dimension 3 as well.}.
It is composed of two rigid bodies $A$ (the square) and $B$ (the triangle); placing three
cad constraints,
a {\bf point-line coincidence}, {\bf line-line perpendicular} and {\bf point-line distance},
results in a {\em generically minimally rigid} (Definition \ref{genMinRigid}) framework.

\begin{figure}
\subfigure[The framework is composed of two rigid bodies $A$ (square) and $B$ (triangle).]{\begin{minipage}[b]{.3\linewidth}\centering\includegraphics[scale=.5]
{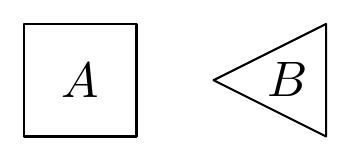}\end{minipage}\label{fig:simple2DexampleBodies}}
\hfil
\subfigure[Constraint 1 specifies a {\bf point-line coincidence} that requires point $p_{B_1}$ on the triangle to lie on line $\ell_{A_1}$ on the square.]{\begin{minipage}[b]{.3\linewidth}\centering\includegraphics[scale=.5]
{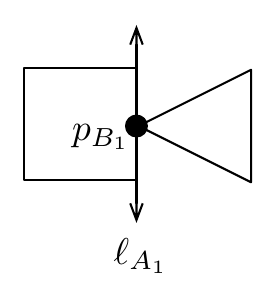}\end{minipage}\label{fig:simple2DexampleC1}}
\hfil
\subfigure[After the specification of Constraint 1, the framework allows internal motion with 2 degrees of freedom.]{\begin{minipage}[b]{.3\linewidth}\centering\includegraphics[scale=.5]
{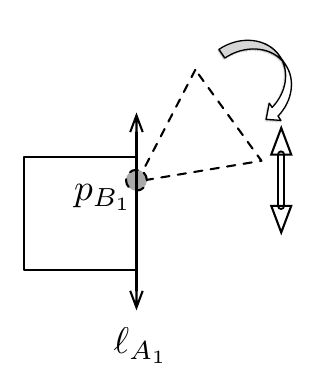}\end{minipage}\label{fig:simple2DexampleC1motion}}\\
\subfigure[Constraint 2 specifies a {\bf line-line perpendicular} constraint that requires line $\ell_{B_2}$ on the triangle to be perpendicular to line $\ell_{A_2}$ on the square.]{\begin{minipage}[b]{.3\linewidth}\centering\includegraphics[scale=.5]
{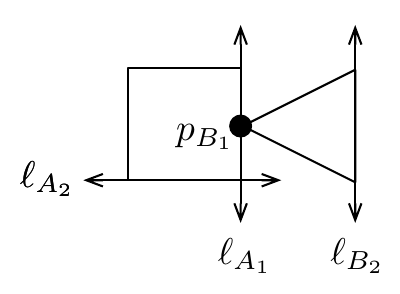}\end{minipage}\label{fig:simple2DexampleC2}}
\hfil
\subfigure[After the specification of Constraints 1 and 2, the framework allows internal motion with 1 degree of freedom.]{\begin{minipage}[b]{.3\linewidth}\centering\includegraphics[scale=.5]
{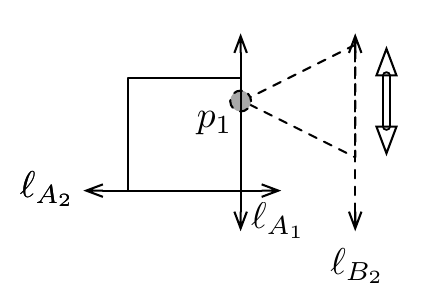}\end{minipage}\label{fig:simple2DexampleC2motion}}
\hfil
\subfigure[Constraint 3 specifies a {\bf point-line distance} constraint that requires point $p_{B_3}$ on the triangle to be a fixed distance from line $\ell_{A_3}$ on the square. The resulting framework is minimally rigid.]{\begin{minipage}[b]{.3\linewidth}\centering\includegraphics[scale=.5]
{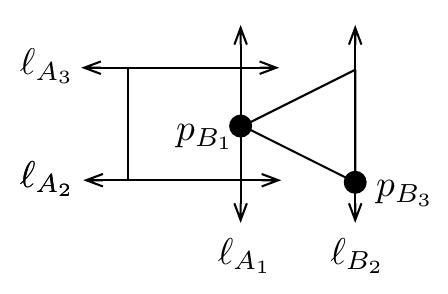}\end{minipage}\label{fig:simple2DexampleGenericC3}}
\caption{A generically minimally rigid 2D body-and-cad framework.}
\label{fig:simple2Dexample}
\end{figure}

\begin{figure}
\subfigure[Constraint I specifies a {\bf line-line distance} constraint that requires line $\ell_{B_1}$ on the triangle to be a fixed distance from line $\ell_{A_1}$ on the square.]{\begin{minipage}[b]{.3\linewidth}\centering\includegraphics[scale=.5]{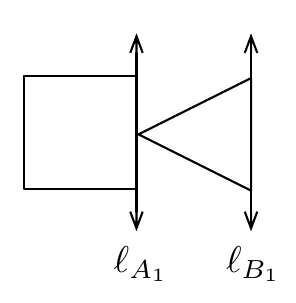}\end{minipage}\label{fig:polyExample2DC12}}
\hfil
\subfigure[After the specification of Constraint I, the framework allows internal motion with 1 degree of freedom.]{\begin{minipage}[b]{.3\linewidth}\centering\includegraphics[scale=.5]{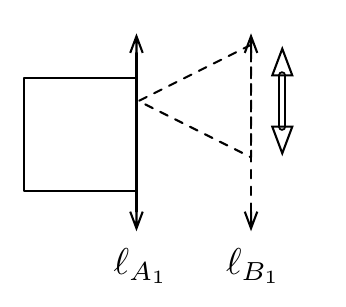}\end{minipage}\label{fig:polyExample2DC12motion}}
\hfil
\subfigure[Constraint II specifies a {\bf point-line distance} constraint that requires point $p_{B_3}$ on the triangle to be a fixed distance from line $\ell_{A_3}$ on the square. The resulting framework is minimally rigid.]{\begin{minipage}[b]{.3\linewidth}\centering\includegraphics[scale=.5]{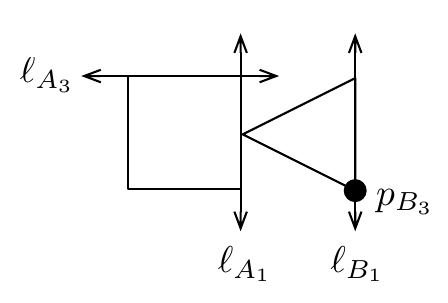}\end{minipage}\label{fig:polyExample2DC3generic}}
\caption{A 2D body-and-cad framework with different geometric constraints but the same rigidity matrix.}
\label{fig:poly2Dexample}
\end{figure}

Now consider the framework in Figure \ref{fig:poly2Dexample}; it is composed of the same two
rigid bodies as in Figure \ref{fig:simple2Dexample},
but only includes two cad constraints, a {\bf line-line distance} and {\bf point-line distance}.
Because Constraint I ({\bf line-line distance}) is equivalent to Constraints 1 and 2
({\bf point-line coincidence} and {\bf line-line perpendicular}), this system of constraints is %
equivalent to the figure in Figure \ref{fig:simple2Dexample}
and is {\em generically minimally rigid}.

\subsubsection{Combinatorial model.}
As discussed above (refer to \figref{poly2Dexample}), a single CAD constraint
may impose restrictions on multiple degrees of freedom in a framework.  Thus, the combinatorial
representation for a body-and-cad framework is a \emph{cad graph}
which has a vertex for each body
and a labeled edge for each cad constraint (refer to the top of Figure \ref{fig:2DexComb}
for the cad graphs for the examples in Figures \ref{fig:simple2Dexample} and \ref{fig:poly2Dexample}).
\begin{figure}
\subfigure[The {\em cad graph} (top) for {\bf Example 1} from Figure \ref{fig:simple2Dexample} has three edges between the vertices for bodies $A$ and $B$. The {\bf line-line perpendicular} constraint is associated to a primitive angular constraint, so the {\em primitive cad graph} (bottom) has one red edge.] {\begin{minipage}[b]{.48\linewidth}\centering\includegraphics[width=.9\linewidth]
{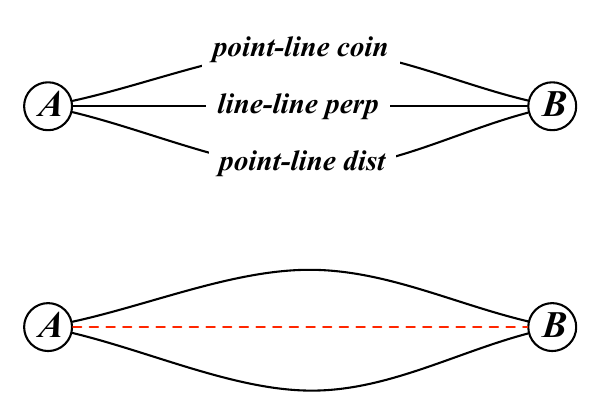}\end{minipage}\label{fig:simpleExample2Dcombinatorics}}
\hfil
\subfigure[The {\em cad graph} (top) for {\bf Example 2} from Figure \ref{fig:poly2Dexample} has two edges between the vertices for bodies $A$ and $B$. Because the {\bf line-line distance} constraint is associated to two primitive cad constraints (one blind, one angular), the {\em primitive cad graph} (bottom) has three edges.]{\begin{minipage}[b]{.48\linewidth}\centering\includegraphics[width=.9\linewidth] {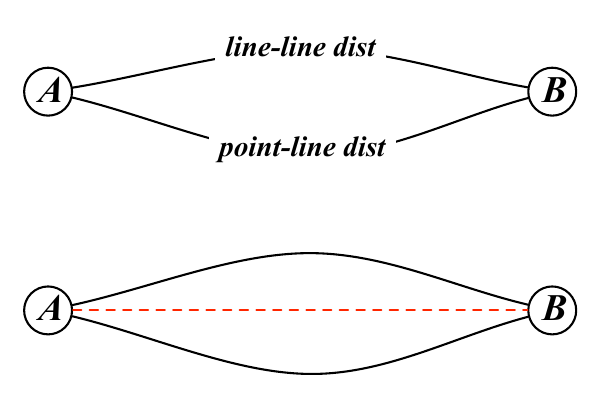}\end{minipage}\label{fig:polyExample2Dcombinatorics}}
\caption{Combinatorics of body-and-cad frameworks in Figures \ref{fig:poly2Dexample}
and \ref{fig:poly2Dexample}: both {\em cad graphs} (top) are associated to the same {\em primitive cad graph} (bottom).}
\label{fig:2DexComb}
\end{figure}

Associated to each cad graph is a  {\em primitive cad graph}, which
is a {\em bi-colored graph} $G = (V, E = R \sqcup B)$
on vertex set $[n]=\{1, \ldots, n\}$
with a vertex for each rigid body,
a red edge (in $R$) for each angular constraint,
and a black edge (in $B$) for each blind constraint.
Notice that the primitive
cad graphs in Figures \ref{fig:simple2Dexample} and \ref{fig:poly2Dexample} are the same
(refer to the bottom of Figure \ref{fig:2DexComb}).

In the rest of this paper, we will work with primitive cad graphs.

\subsubsection{The rigidity matrix and infinitesimal rigidity.}
As is standard in the field, we will linearize the
geometric constraint equations and consider \emph{infinitesimal rigidity}.  Here, the
core object of study is a \emph{rigidity matrix} (derived in \cite{halleretal}),
whose kernel consists of the infinitesimal motions of the framework.

To describe body-and-cad rigidity matrices combinatorially, we use the
following concept\footnote{The $[a,b]$-frame
defined here is equivalent to the $(a+b,a)$-frame defined in \cite{stjohnSidman}.}.
\begin{definition}
For integers $a,b$, let $k=a+b$.
We define an {\em $[a,b]$-frame $G(\bp)$} to be a bi-colored graph $G = (V, E = R \sqcup B)$ with $kn-k$ edges,
along with a function $\bp:E \to \RR^k$.
The function $\bp$ labels each edge with a $k$-vector, which
\ASJcomment{changed to match what we do later; put a coords at beginnign}
is zero in the last $b$ entries if the edge is in $R.$
The {\em generic $[a,b]$-frame $G(\bx)$} has formal indeterminates replacing the nonzero coordinates of
$\bp.$
\end{definition}

We define the rigidity matrix in terms of $[a,b]$-frames.  We first fix some ordering on the edges of $G$.
\JScomment{I think we do need to say that the edges of $G$ are ordered, but it really doesn't matter what it
is, so I am loathe to specify it explicitly.}
\begin{definition}\label{def: rigidityMatrix}
The {\em rigidity matrix} $M(G(\bp))$ of an $[a,b]$-frame $G(\bp)$ is a matrix that has $k$
columns for each vertex $i$ and one row for each edge of $G.$ \JScomment{We had a ``loop" in here, which I deleted.  I think that was a vestige of an earlier version.}
In the row corresponding to
an edge $e$ with endpoints $i$ and $j$ (where $i < j$),
we have $\bp(e)$ in the columns corresponding to $i$, $-\bp(e)$ in the columns corresponding to $j$,
and zeroes in all other entries.  Order the rows of the rigidity matrix in the order of the edges of $G.$
\end{definition}
\ASJcomment{modified to explicitly refer to generic frame}
\begin{definition}\label{genMinRigid}
We say that an $[a,b]$-frame $G(\bp)$ is {\em generically minimally rigid} if the associated
generic $[a,b]$-frame $G(\bx)$ has a rigidity
matrix $M(G(\bx))$ with rank $kn-k.$
\end{definition}

The generic rigidity matrix for the example in Figure \ref{fig:poly2Dexample}  is shown below.
$$\bordermatrix{& A_1 &A_2 &A_3 &B_1 &B_2 &B_3 \cr
\text{line-line distance (blind part)} & x_1 & x_2 & x_3 & -x_1 & -x_2 & -x_3 \cr
\color[rgb]{1,0,0}
\text{line-line distance (angular part)}  & \color[rgb]{1,0,0} y_1 & \color[rgb]{1,0,0} 0 &\color[rgb]{1,0,0} 0 & \color[rgb]{1,0,0} -y_1 &\color[rgb]{1,0,0}  0 & \color[rgb]{1,0,0}0  \cr
\text{point-line distance (non-generic)} & z_1 & z_2& z_3 & -z_1 & -z_2 & -z_3 }$$

It has 3 columns for each body, corresponding to
the one rotational and two translational degrees of freedom for infinitesimal rigid
body motion in the plane. We order the columns so that they are in groups of 3, with
translational
components last: column $A_1$ corresponds to the rotational component, and
columns $A_2$ and $A_3$ to the translational components, with body $B$'s columns ordered analogously.
There is a row for
each  primitive constraint; notice that that row for the primitive angular constraint associated to the {\bf line-line distance} constraint has
zeroes (highlighted in red) in the columns corresponding to the translational degrees of freedom.

\subsection{Minimally rigid graphs}\label{sec:combDef}
A result from \cite{stjohnSidman}
gives a combinatorial characterization of generic minimal rigidity for 2D body-and-cad frameworks
(with $[1,2]$-frames) and, omitting point-point coincidence constraints, for 3D body-and-cad
frameworks (with $[3,3]$-frames):
\begin{theorem}\theolab{bodyCadChar}
An $[a,b]$-frame with bi-colored graph $G = (V, E = R \sqcup B)$
is generically minimally rigid if and only if $\exists B' \subseteq B$ such that:
\begin{itemize}
\item $(V,R \cup B')$ is the edge-disjoint union of $a$ trees, and
\item $(V,B \setminus B')$ is the edge-disjoint union of $b$ trees
\end{itemize}
\end{theorem}

\theoref{bodyCadChar} can also be stated in terms of \emph{hereditary sparsity}, which we now recall.
A multigraph $G = (V,E)$ is {\em $(k,\ell)$-sparse}
if every subset of $n'$ vertices spans
at most $kn'-\ell$ edges; if in addition, $G$ has exactly $kn-\ell$ edges, it is called {\em $(k,\ell)$-tight}.
For brevity, $(k,\ell)$-tight graphs will be called {\em $(k,\ell)$-graphs}.
A subset of vertices of $G$ that induces a $(k,\ell)$-graph is a {\em $(k,\ell)$-block}.
When $\ell\in [0,2k)$, $(k,\ell)$-graphs are the bases of the {\em $(k,\ell)$-matroid} \cite{leeStreinu}.

\begin{definition}
Let $G = (V, E = B \disjointUnion R)$ be a bi-colored graph, $a,b$ be positive integers, and $k = a+b$.
Then $G$ is {\em $[a,b]$-sparse}\footnote{As we
will rely on both concepts of sparsity, we draw the
reader's attention to the use of {\bf parentheses} to denote the parametrized
$(k,\ell)$-sparsity of simple graphs (appearing for classical bar-and-joint and body-and-bar rigidity)
and the use of {\bf square brackets} to denote the parametrized $[a,b]$-sparsity
counts of bi-colored graphs (introduced for body-and-cad rigidity).
}
if $\exists B' \subseteq B$ such that:
\begin{itemize}
\item $(V,R \cup B')$ is $(a,a)$-sparse, and
\item $(V,B \setminus B')$ is $(b,b)$-sparse
\end{itemize}
Additionally, if $G$ has exactly $kn-k$ total edges, then
$G$ is {\em $[a,b]$-tight} and referred to as an {\em $[a,b]$-graph}.
\end{definition}
A subset of vertices of an $[a,b]$-sparse graph that induces an $[a,b]$-graph is
an \emph{$[a,b]$-block}.

The Nash-Williams and Tutte Theorem \cite{tutte,nashWilliams} states that
$G$ is a $(k,k)$-graph if and only if $G$ is the edge-disjoint union
of $k$ spanning trees.
Therefore,
the $[a,b]$-graph
property
is equivalent to there existing $B' \subseteq B$ such that $(V, R \cup B')$ is the edge-disjoint union
of $a$ spanning trees and $(V, B \setminus B')$ is the edge-disjoint union of
$b$ spanning trees.

That this class is matroidal follows from the Matroid Union Theorem \cite[Prop. 7.6.14]{B86}.
Therefore, for an $[a,b]$-sparse graph $G=(V,E)$ and an edge $e$ not in $E$, we will say that
$e$ is {\em independent} of $G$ if $G+e$ is also
$[a,b]$-sparse and {\em dependent} otherwise.

\subsection{Non-genericity}
While the rank of a generically minimally rigid framework is $kn-k$ for almost all realizations, there are realizations for which the rank drops.  These correspond
to non-generic realizations, or {\em special positions}, of the generically minimally rigid graph.

\begin{figure}[bth]
\subfigure[As in Figure \ref{fig:poly2Dexample}, we specify a {\bf point-line distance} constraint that requires point $p_{B_3}$ on the triangle to be a fixed distance from line $\ell_{A_3}$ on the square. The resulting (generic) framework is minimally rigid.]{\begin{minipage}[b]{.3\linewidth}\centering\includegraphics[scale=.5]{polyExample2DC3generic}\end{minipage}\label{fig:polyExample2DC3genericAgain}}
\hfil
\subfigure[An alternative choice of the line on body $A$ involved in Constraint II.]{\begin{minipage}[b]{.3\linewidth}\centering\includegraphics[scale=.5]
{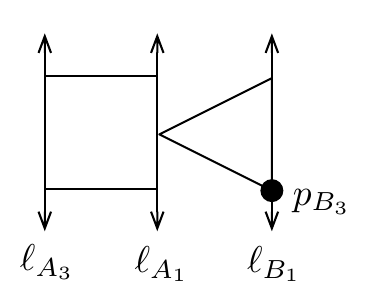}\end{minipage}\label{fig:polyExample2DC3nongeneric}}
\hfil
\subfigure[The resulting flexible non-generic framework admits an internal motion with 1 degree of freedom.]{\begin{minipage}[b]{.3\linewidth}\centering\includegraphics[scale=.5]
{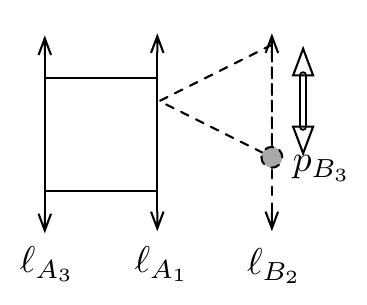}\end{minipage}\label{fig:simple2DexampleNonGenericC3motion}}
\caption{A {\bf flexible special position} of a generically minimally rigid 2D body-and-cad framework.}
\label{fig:simple2DexampleNongeneric}
\end{figure}
In Figure \ref{fig:simple2DexampleNongeneric} we consider
a special position of the framework specified by the graph from Figure \ref{fig:poly2Dexample}.
By changing the choice of placement for the line on body $A$ in the final constraint, we have constructed a non-generic realization of the
framework which has a a consistent dependency and remains flexible.

Its rigidity matrix is shown below and
contains a dependency (the third row is the sum of the first two),
causing its pure condition to vanish.
$$\bordermatrix{& A_1 &A_2 &A_3 &B_1 &B_2 &B_3 \cr
\text{point-line coincidence} & 1& 0 & -1  & -1 & 0 & 1  \cr
\color[rgb]{1,0,0} \text{line-line perpendicular} & \color[rgb]{1,0,0}-1 & \color[rgb]{1,0,0} 0 & \color[rgb]{1,0,0} 0  &\color[rgb]{1,0,0} 1 & \color[rgb]{1,0,0} 0 & \color[rgb]{1,0,0} 0  \cr
\text{point-line distance (non-generic)}& 0 & 0 & -1  & 0 & 0 & 1}$$

\section{Structure of the pure condition}\seclab{PC}
In this section, we review the definition of a polynomial called the pure condition
that is associated to a body-and-cad framework.  If
$G$ is a minimally rigid graph and $\bp$ is generic, the
kernel of $M(G(\bp))$ contains exactly
the space of trivial infinitesimal motions of of $G(\bp)$, corresponding
to rigid-body motions of the entire framework as a single unit.  To remove these, we choose some
$i$ with $1\leq i \leq n,$ and construct
the \emph{standard tie-down} at body $i$
by  appending   to $M(G(\bp))$ a $k \times kn$ matrix whose
only nonzero entries are given by the identity matrix in the $k$ columns associated to body $i$.\footnote{The
notion of tie-downs can be substantially generalized to generic tie-downs of any
$[a,b]$-sparse graph.  For our purposes, this would only complicate the notation, so we
restrict our attention to standard tie-downs.}
We denote the rigidity matrix of $G(\bp)$ with a tie-down
by $M_T(G(\bp))$.

\ASJcomment{do we need a figure here? Like, one to show trees versus fans and associated pure condition?}
\begin{definition}
The \emph{pure condition} $P_G$ of a tied down $[a,b]$-graph $G$ is the determinant of
$M_T(G(\bx))$.
\end{definition}
\ASJcomment{do we want to say something about P_G being a polynomial}

The pure condition depends, a priori, on the choice of tie-down of $G$.  We will show that as in
the body-bar setting of \cite{whiteWhiteley}, this dependence can be removed.
\begin{theorem}\theolab{pure-condition-properties}
The pure condition of a tied down $[a,b]$-graph $G$ is non-zero for \emph{any} choice\footnote{The standard tie-down pins $k$ coordinates of one body, but we can also choose
``generic'' tie-downs that pin $k$ coordinates chosen from different bodies.}
of the tie down and has the form
$P_G = T_G\cdot C_G$, where $T_G$ depends on the tie down and $C_G$ is independent of the tie down.
\end{theorem}
This result follows from two combinatorial formulas for the pure condition that we
derive below.  Since the proofs are relatively standard and similar to what
can be found in \cite{whiteWhiteley,W88}, we develop them quickly.

\subsubsection*{Tree decompositions.}
Let $G$ be an $[a,b]$-graph.  A \emph{tree decomposition} $\mathcal{T} = (A_1,\ldots,A_a,T_1,\ldots,T_b)$
of $G$ is a partition of the edges into $k$ spanning trees such that all the red edges are in
the sets $A_i$.  It follows from the definitions in Section \ref{sec:combDef}
that there is a tree decomposition of a bi-colored
graph if and only if it is an $[a,b]$-graph.

To make the connection to the pure condition we %
define the \emph{tree decomposition monomial} $\{\mathcal{T}\}$ to be
\[
\{\mathcal{T}\} :=	\prod_{e \in A_i, \ i\in [a]}\bx(e)_i
\prod_{e \in T_j, \ j\in [b]}\bx(e)_{a + j}
\]

\noindent The tree monomials are precisely the monomials appearing in $P_G.$ We first give a concrete example
before stating the general theorem.

\ASJcomment{moved example before theorem}
\begin{ex}\label{ex:22}
\JScomment{We may ultimately want a more complicated example, but this one already does display many of the features that we encounter and it is small.}
\ASJcomment{We have mixed up which coordinates are 0. In Section 2, we have the first b coords as 0, but I think we keep using the last b as 0... I lean towards changing Section 2 to be the last b coords as 0. I'm going to do that now...}
Let $G$ be the $[2,2]$-graph on vertices 1 and 2 with two red edges and two black edges depicted in
Figure \ref{fig:basic22}.
\begin{figure}[htb]
\centering\includegraphics[scale=.5]{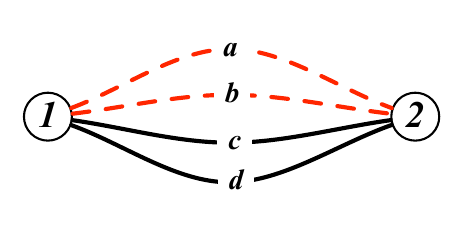}
\label{fig:basic22}
\caption{A $[2,2]$-tight graph on vertices $1$ and $2$ with two red (dashed) edges $a$ and
$b$ and two black (solid) edges $c$ and $d$.}
\end{figure}

\noindent The tied down rigidity matrix is
\[
\begin{pmatrix}
a_1 & a_2 & 0 & 0 & -a_1 & -a_2 & 0 & 0\\
b_1 & b_2 & 0 & 0 & -b_1 & -b_2 & 0 & 0\\
c_1 & c_2 & c_3 & c_4 & -c_1 & -c_2 & -c_3 & -c_4 \\
d_1 & d_2 & d_3 & d_4 & -d_1 &- d_2 &- d_3 &- d_4\\
1 & 0 & 0 & 0 & 0 & 0 & 0 & 0\\
0 & 1 & 0 & 0 & 0 & 0 & 0 & 0\\
0 & 0 & 1 & 0 & 0 & 0 & 0 & 0\\
0 & 0 & 0 & 1& 0 & 0 & 0 & 0\\
\end{pmatrix}
\]
There are 4 tree decompositions:
\[(\{a\},\{b\},\{c\},\{d\}), (\{b\},\{a\},\{c\},\{d\}), (\{a\},\{b\},\{d\},\{c\}), (\{b\},\{a\},\{d\}, \{c\}).\]
The corresponding monomials give the determinant of the rigidity matrix with a standard tie-down: $a_1b_2c_3d_4-a_2b_1c_3d_4 -a_1b_2c_4d_3+a_2b_1c_4d_3$.
\JScomment{Modifications here so that things fit in margin, looked good, and made sense.}

\end{ex}

\begin{theorem}\theolab{bodycad-tree-form}
If $G$ is a tied down $[a,b]$-graph
\begin{equation}
\eqlab{ab-trees}
P_G  =  \sum_{\substack{\text{tree decomps.} \\  \mathcal{T}}} \pm\{\mathcal{T}\} \qquad \text{in $\mathbb{R}[\bx]$}.
\end{equation}
The signs can be determined using the definition of the determinant of a matrix.  The precise formula for the signs is technical and not needed in what follows, so we do not describe it here.  %
\end{theorem}

\begin{proof}%
We sketch a proof that follows the proof of Theorem 2.18 in \cite{whiteWhiteley}.  First we reorder the
columns of $M_T(G(\bx))$ by the coordinates of $\bp(e)$ so that we have the $n$ first coordinates, followed
by the $n$ second coordinates, etc.  In doing this we can see that
if we ignore the last $k$ rows corresponding to the tie down, each successive collection of $n$ columns is
just an incidence matrix for the directed multigraph $G$ whose rows are the edges of $G$, and whose columns correspond to the vertices.

We expand the determinant of $M_T(G(\bx))$ along these successive groups of $n$ columns.  To do this,
in each set of $n$ columns,
we need to choose $n$ rows.  Since this set of $n$ columns is an incidence matrix, if this subdeterminant is nonzero, then these rows must
correspond to a spanning tree plus a row corresponding to the tie down.  Moreover, this subdeterminant
is actually a monomial as it is possible to expand it row by row, choosing the tie down as the first row, choosing
an edge incident to the tied-down vertex as the next row,
and choosing successive rows by taking an edge that was
adjacent to an edge already chosen.   Proceeding in this way, each row that we choose has only one
nonzero entry.

The resulting $n \times n$ determinant
is multiplied by $k$ others, all chosen in the same way, to obtain a monomial.  Since we cannot permit any row to appear twice
in such a product, this product of $k$ $n \times n$ determinants corresponds to a decomposition of
$G$ into $k$ edge-disjoint spanning trees.

Moreover, we argue that such a product is nonzero if and only if the $k$ trees form an $[a,b]$-tree decomposition.
To see this, note that within the last $b$ groups of $n$ columns, the rows corresponding to red edges have only
zeroes as entries.  So, all of the red edges are in the $a$ trees corresponding to the first $a$ coordinates.

\end{proof}

In particular, since the tree decompositions are independent of the tie-down, $P_G$ is
determined up to sign.
\JScomment{I think that if we meditate on it, it's clear that all of the signs undergo the same change.}
We define the \emph{critical factor}\footnote{The name comes
from the setting of general tie-downs, where $\pm 1$ is
replaced with the determinant of a $k\times k$ matrix.} $C_G$ to be $P_G$ with respect to the
tie-down that pins vertex $1$, to establish a convention.

\subsubsection*{Fans.}
An alternative combinatorial formula generalizes one based on the $k$-fans of
\cite{whiteWhiteley}.  Fix the ordering and base orientation of the edges as in Definition
\ref{def: rigidityMatrix}, where edges are oriented from $i$ to $j$ if $i<j$.
The ordering on the edges induces an ordering on subsets.  %
For a tied-down $[a,b]$-graph, an
$[a,b]$-fan $\mathcal{F}$ is an orientation of the edges such that: (a) all edges incident
to the tied-down vertex $i$ are oriented toward $i$; (b) at all other vertices,
exactly $k$ edges are oriented out, at most $a$ of which are red.  Define $F_j$
to be the set of edges oriented out of vertex $j$.  The sign $\varepsilon(\mathcal{F})$
is $(-1)^t$, where $t$ is the number of edges oriented opposite to the base orientation
by $\mathcal{F}$.  We define the \emph{fan monomial} $[\mathcal{F}]$ by
\[
[\mathcal{F}] :=\operatorname{sgn}(\sigma)\varepsilon(\mathcal{F})\prod_{j\in [n]\setminus \{i\}} [F_i]
\]
where $[F_j]$ is the \emph{bracket} associated to the (ordered)
set of vectors $\bx(e)$ for $e\in F_i$ and $\sigma$ is the permutation of the
edges that puts them in the order $F_1,F_2,\ldots, F_n$.  The brackets
can be thought of either as elements of the
homogeneous coordinate ring of the Grassmannian, or simply as the
determinant of the $k\times k$ matrix with columns $\bx(e)$.

It is not hard to check the following, which is a generalization of Proposition 2.12 in \cite{whiteWhiteley}.
\begin{theorem}\theolab{bodycad-det-form}
For a tied-down $[a,b]$-graph, $G$, we have
\begin{equation}\eqlab{ab-kfans}
P_G = \sum_{\substack{\text{$[a,b]$-fans} \\ \mathcal{F}}} [\mathcal{F}].
\end{equation}
\end{theorem}

\begin{ex}[name={Example \ref{ex:22}, continued}]
If $G$ is the graph on two vertices in Figure \ref{fig:basic22}, and we tie down vertex 1, we see that there is a unique $[2,2]$-fan with all four edges pointing from vertex 2 to vertex 1.  Hence, $P_G=[abcd],$ and even in this very simple example, we can see how expressing $P_G$ in terms of brackets greatly simplifies notation.
\end{ex}

\section{Factors of the pure condition and factor graphs}\seclab{factors}
In this section, we investigate the structure of factors of the
(critical factor of the) pure condition of an $[a,b]$-graph.  The goal
is to identify these factors graph-theoretically.  This is more complicated
than in the body-and-bar setting where irreducible factors correspond to
rigid subgraphs.  In our case we will have some factors that correspond
to rigid $[a,b]$-subgraphs, some that correspond to rigid $(a,a)$ or
$(b,b)$-subgraphs, and others that cannot be interpreted as the pure
condition of any rigid subgraph.

For the rest of this section
$G$ will refer to an $[a,b]$-graph, so that $C_G$ is well-defined.
To further understand the structure of the pure condition, we may
contract subgraphs.
If $H$ is a subgraph of $G$, we write $G/H$ to denote the graph obtained by
deleting the edges of $H$ and identifying the vertices of $H$.

\begin{definition}
Let $G$ be an $[a,b]$-graph with pure condition $C_G = fg$.  We say that the \emph{edge support}
of the factor $f$ is the set $E_f$ of edges $e$ in $G$ such that some variable of $\bx(e)$
is in $f$.
\end{definition}
The supports of distinct factors define edge-disjoint subgraphs of $G$.

\begin{theorem}\theolab{factor-graphs-exist}
Let $C_G = fg$.  Then the edge supports of $f$ and $g$ are disjoint and partition $E(G)$.  Moreover,
every monomial of a factor contains exactly one coordinate from each edge in its support.
\end{theorem}
\begin{proof}
From \eqref{ab-trees} we can see that every monomial of $C_G$ contains exactly one coordinate from each edge.
If the coordinates of $\bx(e)$ were split between two distinct factors, then their product would have
terms divisible by more than one coordinate of $\bx(e),$ which would be a contradiction.
\end{proof}

\theoref{factor-graphs-exist} allows us to make the following definition.
\begin{definition}\label{proper}
A bi-colored graph $H$ is a \emph{factor graph} of $G$ if $H$ is a minor\footnote{Graph minors are
obtained by deleting edges and vertices and contracting edges.} of $G$ and $C_G = hf$, with
the factor $h$ supported on $H$.  The factor graph $H$ is \emph{proper}, if $h = C_H$ for some
tie down of $H$.  If $H$ is all red we abuse notation and consider it as an $(a,a)$-graph.  If $H$ is a black $(b,b)$-graph whose vertices are contained in a red $(a,a)$-graph, then we consider $H$ as a $(b,b)$-graph labeled by the last $b$ coordinates of the edge vectors.
\JScomment{modified this so that $J$ can be a proper factor graph in Theorem 8}
\end{definition}

\JScomment{Added a paragraph here.}
\begin{ex}[name={Example \ref{ex:22}, continued again}]\label{ex: proper}
Define $H$ to be the subgraph of $G$ containing the two red edges.  Then $H$ is a proper factor graph as it is the pure condition of this graph as a $(2,2)$-graph.  The black factor graph consisting of the other two edges is also proper, though its edges are now labeled by the 2-vectors $(c_3,c_4)$ and $(d_3,d_4).$

Here we note two key differences between $[a,b]$-graphs and $(k,k)$-graphs and discuss their consequences.  First, in the pure condition of a $(k,k)$-graph, if $e$ is an edge, every coordinate of $\bx(e)$ appears in some monomial in the pure condition.  In this example, $c_1, c_2, d_1,$ and $d_2$ do not divide any monomial.   Second, in the $(k,k)$-graph setting, each irreducible factor is the pure condition
of a $(k,k)$-graph. In our example, $P_G$ factors as $(a_1b_2-a_2b_1)(c_3d_4-c_4d_3),$ yet neither factor is the pure condition of an $[a,b]$-subgraph of $G.$

\end{ex}

This next theorem gives a first indication of why it is convenient to define factor graphs in terms of minors
instead of, e.g., subgraphs or subset of edges, as we see that $C_G$ may factor as the product of the
pure condition of a proper subgraph $H$
and the pure condition of $G/H,$ the minor obtained by contracting $H$.
The geometric meaning of \theoref{kk-factor} is that, if $H$ corresponds to a rigid
sub-framework (i.e., an $[a,b]$-subgraph), we can split the special positions of $G$ into two types:
(i) special positions of $H$, which induce special positions of $G$; (ii) special positions of
$G$ where $H$ is placed generically.
In case (ii), $H$ can be replaced by a single rigid body, so these special positions
are special positions of $G/H$.
\theoref{kk-factor} is a small modification of
\cite[Theorem 4.12]{whiteWhiteley} to body-and-cad.

\begin{theorem}\theolab{kk-factor}
Let $G$ have a proper subgraph $H$ that is an $[a,b]$-block.  Then $H$ and $G/H$ are both proper
factor graphs and $C_G = C_H\cdot C_{G/H}$.
\end{theorem}
The proof requires the following standard lemma.
\begin{lemma}\lemlab{kk-contract}
Let $G$ be a $(k,k)$-graph, and let $H$ be a proper block.  Then $G/H$ is also
a $(k,k)$-graph.
\end{lemma}
\begin{proof}
Fix a tree decompisition $T_1,\ldots,T_k$ of $G$, and let $n'$ be the number of
vertices of $H$.  Since $H$ has $m'=kn'-k$ edges, all of which are covered
by one of the trees, $|T_i\cap H| = n'-1$ for every $i$.  Thus, contracting
$H$ involves contracting a subtree of each of the $T_i$.
Since contracting a connected subtree of any tree $T$ produces a smaller tree
$T'$, contracting $H$ produces a set of $k$ trees that cover the edges of $G/H$.
By the Tutte-Nash-Williams Theorem $G/H$ is also a $(k,k)$-graph.
\end{proof}
\begin{proof}[Proof of \theoref{kk-factor}]
Since $H$ is an $[a,b]$-graph, it has a pure condition $C_H$.
\lemref{kk-contract} then implies that $G/H$ is also
an $[a,b]$-graph, so $C_{G/H}$ is defined as well.
If both $C_H$ and $C_{G/H}$ divide $C_G$, then
$C_G = C_G\cdot C_{G/H}$, since the supports of $H$ and $G/H$
are disjoint and partition the edges of $G$.

To see that $C_H$ and $C_{G/H}$ divide $C_G,$ we will use \eqref{ab-trees}, and a small
elaboration of the proof of \lemref{kk-contract}.  Pick tree decompositons
$\mathcal{T}_H$ and $\mathcal{T}_{G/H}$ of $H$ and $G/H$.
Suppose that the union $T$ of $T^H_1$ and $T^{G/H}_1$ contains a cycle in $G$.
Since $T^H_1$ is a tree, this cycle is not contained entirely in $H$.
This implies that $T/T^H_1 = T^{G/H}_1$  is not acyclic, which is a contradiction.

Because $\mathcal{T}_H$ and $\mathcal{T}_{G/H}$ were arbitrary, a tree decompisiton
of $G$ is determined by picking one of $H$ and $G/H$ independently.  The desired
result then follows from \eqref{ab-trees}.
\end{proof}
\color{black}

It is not too hard to see that for a body-and-bar framework every factor graph is proper.
The next few theorems indicate that body-and-cad is more complicated than body-and-bar.

\begin{theorem}\theolab{aa-factor}
Let $G$ be an $[a,b]$-graph with a subgraph $H$ that contains only red edges and is an $(a,a)$-block.  Then
$H$ is a proper factor graph of $G$.
\end{theorem}
\begin{proof}
Because $H$ is a completely red $(a,a)$-block, any tree decomposition
decomposition of $G$ induces a decomposition of $H$
into $a$ edge-disjoint spanning trees.  The trees covering $H$ can be chosen independently of
the trees covering the remaining edges, so the theorem follows from
\eqref{ab-trees}.
\end{proof}
Figures \ref{fig:example21recurse} and \ref{fig:factorAlgEx11} contain two examples with red $(a,a)$-blocks appearing as proper factor
graphs.  The geometric interpretation
of \theoref{aa-factor} is that as a sub-framework $H$ is not rigid, but \emph{angularly rigid}:
only translational
degrees of freedom are permitted between the vertices spanned by $H$.
In light of this, we define the following.
\begin{definition}
Let $G$ be $[a,b]$-sparse. If $H$ is a subgraph that contains only red edges and is
an $(a,a)$-block, we say that it is an \emph{angular block}.  If $G$ is an $[a,b]$-graph,
then an angular block is an \emph{angular factor graph}.
\end{definition}
A key difference between the situation in \theoref{kk-factor} and that of \theoref{aa-factor} is that
we cannot simply contract $H$ and obtain another proper factor graph.  Indeed, as the example
in \figref{example11improper} shows, $G\setminus H$ may be an \emph{improper} factor graph,
indicating that the distinction is necessary.

So far we have seen four types of proper factor graphs -- $[a,b]$-subgraphs, quotients of $[a,b]$-subgraphs, all red $(a,a)$-subgraphs and all black $(b,b)$-subgraphs (as in Example \ref{ex: proper}).  Before giving the
detailed context in which black $(b,b)$-subgraphs may be proper factor graphs, which is quite technical, we provide the geometric intuition.
\JScomment{Check this}
Refer back to the example in \figref{2DcontextuallyRigidComp}. Suppose that
$G$ contains an angular factor $H$ (parallel constraints in the example).
As previously described, the vertices spanned by $H$ ($A$, $B$ and $C$ in the example) are
``translating bodies'' relative to each other; then, a $(b,b)$-block $J$ on a subset of the vertices
of $H$ ($B$ and $C$ in the example) will, in fact, completely rigidify them.  This will hold
\emph{even if the vertex set of $J$ does not induce an $[a,b]$-block},
which implies that body-and-cad, even in in the plane, allows us to construct frameworks with
generically rigid components that are not rigid as induced subgraphs.
\begin{definition}
Let $G$ be an $[a,b]$-sparse graph with an angular factor graph $H$, and let
$J$ be a $(b,b)$-block on a subset of the vertices of $H$.
If $J$ does not span all the vertices of $H$, then $J$ is a \emph{contextually rigid component}.
\end{definition}
By analogy with the geometric reasoning used to identify the factor graph associated with an $[a,b]$-block,
we want to think of the vertices of $J$ as corresponding to a single rigid body.
While we previously continued by looking for proper factor graphs in $G/J$,
a more complicated construction is required when $J$ does not span an $[a,b]$-block .

\begin{theorem}\theolab{bb-factor}
Let $G$ be an $[a,b]$-graph with an all red $(a,a)$-block $H$ as a subgraph and an
all black $(b,b)$-block $J$ with $V(J) \subset V(H)$.  Tie down a vertex in $J$, fix
an $[a,b]$-tree decomposition $\mathcal{T}$ of $G$, and orient all edges towards this root.
Define $G'$ by removing all of the edges of $H$ whose tails are in $J$ and then contracting $J.$
\LTcomment{Check.} Then,
\begin{enumerate}
\item  $H$ and $J$ are proper factor graphs
\item  $C_G = C_H \cdot C_J \cdot f$, where $f$ is supported on a subgraph $F$ whose edges are the edges of $G$ that are not in $H$ or $J.$
\item  $G'$ is an $[a,b]$-graph with $C_{G'} = hf,$ where $h$ is supported on
the red $(a,a)$-block $H' = H / J$.
\end{enumerate}
\end{theorem}

\begin{proof}
That $H$ is a factor graph of $G$ follows from \theoref{aa-factor}.  To see that
$J$ also is, observe that,
since $V(J)\subset V(H)$, all the edges of $J$ are
covered by trees $T_j$  (and not $A_i$).  Thus, the decomposition of $J$ into $b$ trees
is independent of the rest of any tree decomposition.  Hence,
$C_J$ (the pure condition of a $(b,b)$-graph) divides $C_G$, using the same arguments as above.

To show that $G'$ is an $[a,b]$-graph, we will show that $\mathcal{T}$ restricts to an
$[a,b]$-tree decomposition of $G'.$  In each tree there was a unique directed path from each vertex
in $G$ to the root in $J$.  Since our construction only deletes edges directed out of $J$, there
is still a unique directed path in each tree from each vertex remaining in $G'$ to the root.  If an
undirected cycle were created in this process, it is easy to see that there would have also been
an undirected cycle in the original tree.  Therefore, the restriction of $\mathcal{T}$ to
$\mathcal{T}'$ in $G'$ is an $[a,b]$-tree decomposition.

Finally, we argue that $H' = H / J$ is a red $(a,a)$-block in $G'$.
This will imply that $H'$ is a factor graph of $G'$, yielding the last
assertion of the theorem.  Above, we argued that the $a$ red
trees of $\mathcal{T}$ in $H$ restrict to $a$ red trees in
$\mathcal{T}'$ in $H'$. By Nash-Williams-Tutte, $H'$ is a $(a,a)$-block (of all red edges).
By \theoref{aa-factor}, $H'$ is a proper factor graph
of $G'$ and $C_{G'} = hf$ with $h$ supported on
$H'$.
\end{proof}
\section{Algorithms for detecting dependencies}\seclab{algs}
We now present combinatorial algorithms to aid in detecting dependencies in CAD systems. The first
algorithm, the {\bf $[a,b]$-pebble game}, characterizes $[a,b]$-sparsity and consequently
addresses \emph{generic} body-and-cad rigidity. If the addition of an edge results in
a dependency in a generic realization of the system, the pebble game will find its fundamental circuit.
The second set of {\bf factor graph algorithms} determines the proper and improper \emph{factor graphs} for an
$[a,b]$-graph, which correspond to factors of its pure condition. When the geometry
of a CAD system causes any of these factors to vanish, the system is in a \emph{special position}
and contains a dependency.
\ASJcomment{is that dependency's ``circuit'' the associated factor graph?}

\subsection{Pebble games for $[a,b]$-sparsity}\seclab{pebblegame}\label{pebblegameSec}
Generic minimal rigidity of body-and-cad frameworks is characterized by
$[1,2]$-sparsity in the plane and, omitting point-point coincidence constraints,
$[3,3]$-sparsity in 3D \cite{stjohnSidman}.

\algref{abPG} describes our {\em $[a,b]$-pebble game}
for solving the {\bf Decision}, {\bf Extraction}, {\bf Components}
and {\bf Optimization} algorithmic problems described
in \cite{leeStreinu} for $[a,b]$-sparse graphs as well as detecting
the fundamental circuit of a dependent edge. This algorithm belongs
to a family of pebble game algorithms \cite{theranStreinu,leeStreinu}
that are based on a set of local moves applied to
the edges of a directed graph, where the edges and vertices are covered by
pebbles representing degrees of freedom. The specific preconditions
for each type of move,
which are related to the sparsity parameters,  determine
the sparsity family recognized by the
game.
\begin{algorithm}
\caption{The $[a,b]$-pebble game algorithm.}
\label{alg:abPG}
{\bf Input:} A bi-colored graph $G = (V, E=B \disjointUnion R)$, with black $B$ and red $R$ edges. \\
{\bf Output:} $[a,b]$-sparsity property {\bf tight},
{\bf sparse}, {\bf dependent and contains spanning tight}, or {\bf dependent}. \\
{\em {\bf Setup:} Initialize an empty directed graph $H$ on vertex set $V$.
On each vertex, place $a$ aqua pebbles and $b$ tan pebbles.}\\
{\bf Allowed moves:}
\begin{addmargin}[2em]{0em}%
\noindent\emph{\bf Add red edge $ij$}
[Precondition: $\geq a + 1$ aqua pebbles on $i$ and $j$.]\\
\emph{		  -- Add the new edge, cover it with an aqua pebble from i (there is one by the precondition).\\
-- Orient $ij$ out of $i$.}\\			%
\noindent{\bf Add black edge $ij$}
[Precondition: $\geq a + 1$ aqua pebbles on $i$ and $j$ or
$\geq b + 1$ tan pebbles on $i$ and $j$.]\\
\emph{-- Add the new edge; cover it with a pebble from $i$ using aqua (if there are $a+1$ aqua) or
tan (if there are $b+1$ tan).\\
-- Orient $ij$ out of $i$. }\\
\noindent\emph{\bf Edge reversal}
[Precondition: vertex $j$ has a pebble on it and an in-edge $ij$ covered by the same color.] \\
\emph{-- Reverse the edge by orienting it as $ji$ out of $j$, covering with the pebble from $j$ and returning the
(same color) pebble originally covering $ij$ to $i$.}\\
\noindent\emph{\bf Aqua exchange edge reversal}
[Precondition: vertex $j$ has an aqua pebble on it and a black in-edge $ij$ covered by a tan pebble;
$i$ and $j$ do not belong to the same $(a,a)$-component of aqua pebble covered edges.]\\
\emph{-- Reverse the edge by orienting it as $ji$ out of $j$, covering with the aqua pebble from $j$ and returning the tan pebble originally covering $ij$ to $i$.}\\
\noindent\emph{\bf Tan exchange edge reversal}
[Precondition: vertex $j$ has a tan pebble on it and a black in-edge $ij$ covered by an aqua pebble;
$i$ and $j$ do not belong to the same $(b,b)$-component of tan pebble covered edges.]\\
\emph{-- Reverse the edge by orienting it as $ji$ out of $j$, covering with the tan pebble from $j$ and returning the aqua pebble originally covering $ij$ to $i$.}
\end{addmargin}

\noindent{\bf Method:}
\begin{enumerate}
\item For each edge $e \in E$
\begin{enumerate}
\item If $e$ is black: attempt to collect $b+1$ tan pebbles on its endpoints
with Alg. \ref{alg:abPGfindPebble}.
\item If Alg. \ref{alg:abPGfindPebble} returns {\bf true}: insert $e$ with an {\bf add black edge} move.
\item Else, or if $e$ is red: attempt to collect $a+1$ aqua pebbles on its endpoints
with Alg. \ref{alg:abPGfindPebble}.
\item If Alg. \ref{alg:abPGfindPebble} returns {\bf true}: insert $e$ with an {\bf add black/red edge} move.
\item Else: {\em reject} it and
highlight the edges returned by Alg. \ref{alg:abPGfindPebble}
as the fundamental circuit of the edge (if $e$ is black, this is the union of
both calls to Alg. \ref{alg:abPGfindPebble}).
\end{enumerate}
\item If every edge is added: output {\bf tight} if there are $a+b$ pebbles left and {\bf sparse} otherwise.
\item Else, there were rejected edges: output {\bf dependent and contains spanning tight} if there are $a+b$ pebbles left
and {\bf dependent} otherwise.
\end{enumerate}
\end{algorithm}

\begin{algorithm}
\caption{The subroutine for finding pebbles for the $[a,b]$-pebble game.}
\label{alg:abPGfindPebble}
{\bf Input:} An $[a,b]$-pebble game configuration (a directed bi-colored graph),
an edge $e$, and a desired additional pebble color $c_{e}$ ({\bf aqua} or
{\bf tan}). \\
{\bf Output:} {\bf true} if $a+1$ aqua (if $c_{e}$ is aqua) or $b+1$ tan (if $c_e$ is tan)
pebbles can be collected on the endpoints of $e$ or
{\bf false} otherwise, along with the set of visited edges.\\
{\bf Method:}
\begin{enumerate}
\item Initialize set $F = \emptyset$.
\item Initialize queue $Q = \emptyset$.
Entries of $Q$ will be of the form $(f,c)$, recording
an edge on which to cover with a pebble of color $c$.
\item Set $e.predecessor = \hbox{NIL}$.
\item Enqueue $(e,c_{e})$ into Q.
\item While $Q$ is not empty
\begin{enumerate}
\item Dequeue $(f,c)$.
\item If $f \neq e$ and $f$ is red, continue to the next iteration of the loop.
\item Use the basic pebble game rules to try to collect $a+1$ (if $c$ is {\bf aqua})
or $b+1$ (if $c$ is {\bf tan}) pebbles on the endpoints of $f$;
let $F'$ be the set of edges visited by that search.
\item If the pebbles were collected
\begin{enumerate}
\item Let $g = f$.
\item While $g.predecessor \neq \hbox{NIL}$
\begin{enumerate}
\item Let $d$ be the color of the pebble covering $g$,
$\overline{d}$ be the opposite color, $u$ and $v$ the source and target
of $g$.
\item Collect a pebble of color $\overline{d}$ on $v$ using the basic
pebble game rules with {\bf edge reversal} moves.
\item Perform a {\bf $\overline{d}$ exchange edge reversal}
move to reverse the edge from $v$ to $u$,
covering it with the $\overline{d}$-colored pebble
and releasing a $d$-colored pebble back onto $u$.
\item Set $g = g.predecessor$.
\end{enumerate}
\item Collect $a+1$ (if $c$ is {\bf aqua})
or $b+1$ (if $c$ is {\bf tan}) pebbles on the endpoints of $g (= e)$.
\item Output {\bf true} and  $F \cup F'$.
\end{enumerate}
\item Otherwise
\begin{enumerate}
\item For each edge $g \in F'$ that is not in $F$
\begin{enumerate}
\item Set $g.predecessor = f$; let $\overline{c}$ be the opposite of color $c.$
\item Enqueue $(g, \overline{c})$ into $Q.$
\end {enumerate}
\item Assign $F = F \cup F'$.
\end{enumerate}
\end{enumerate}
\item Output {\bf false} and $F$.
\end{enumerate}
\end{algorithm}
\ASJcomment{the figure needs to be updated, probably to include corresponding Knuth graph}

One way of intuitively understanding the $[a,b]$-pebble game is to imagine separate
$(a,a)$- and $(b,b)$- pebble games played on sets $A$ and $T$,
which partition the current edge set, respectively. The aqua-colored pebbles
track the edges in $A$ as well as the $(a,a)$-sparsity of that partition;
the tan-colored pebbles do the same for $T$ with $(b,b)$-sparsity counts\footnote{We
chose the colors aqua and tan to be associated with $A$ and $T$; in the rigidity
setting, the aqua-covered $A$ partition can be thought of as tracking the angular
degrees of freedom of the system, while the tan-colored $T$ partition tracks
the translational degrees of freedom.}.
The $[a,b]$-pebble game relies on moves that permit black edges to move between $A$
and $T$ in a controlled manner, which corresponds to collecting additional pebbles of certain
colors, using a subroutine described in \algref{abPGfindPebble}.
To find a sequence of these moves, \algref{abPGfindPebble}
specializes Knuth's matroid union algorithm \cite{Knuth:1973:MP:891978}
to the $[a,b]$-sparsity matroid using pebble games.
By enqueuing unvisited edges (in $F' \setminus F$), it uses a
breadth-first approach to find the shortest path (stored with \emph{predecessor} pointers)
to an edge whose pebble color can be exchanged.
To help illustrate the algorithm, \figref{pgExample} shows some
steps of the pebble game played on the primitive cad graph
for \figref{gen3Body2D}.

\begin{figure}[tbh]
\subfigure[The input is a primitive cad graph with 5 black (solid) edges and 1 red (dashed) edge.]{\begin{minipage}[b]{.238\linewidth}\centering\includegraphics[width=\linewidth]{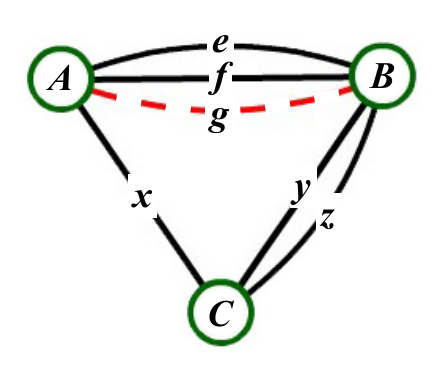}\end{minipage}\label{fig:abPGinput}}
\hfil
\subfigure[The setup stage places $a=1$ aqua (circular) pebble and $b=2$ tan (square) pebbles on each  vertex.]{\begin{minipage}[b]{.238\linewidth}\centering\includegraphics[width=\linewidth]
{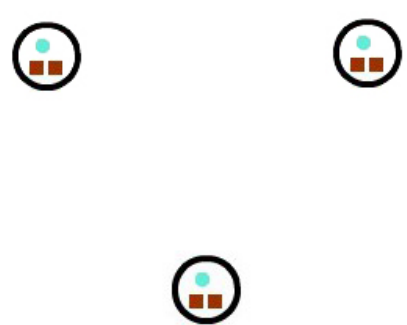}\end{minipage}\label{fig:abPG1}}
\hfil
\subfigure[Since there were at least $b+1 = 3$ tan pebbles on its endpoints, the black edge $e$ is inserted with an {\bf add black edge} move.]{\begin{minipage}[b]{.238\linewidth}\centering\includegraphics[width=\linewidth]
{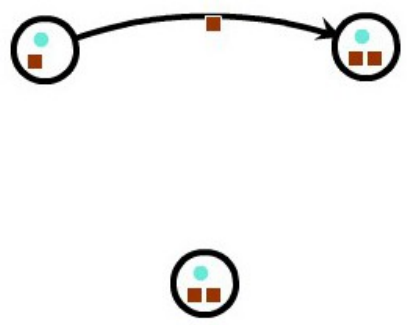}\end{minipage}\label{fig:abPG2}}
\hfil
\subfigure[Another {\bf add black edge} move inserts the edge $f$. While the direction is arbitrarily chosen, a pebble from the source is used to cover the edge.]{\begin{minipage}[b]{.238\linewidth}\centering\includegraphics[width=\linewidth]
{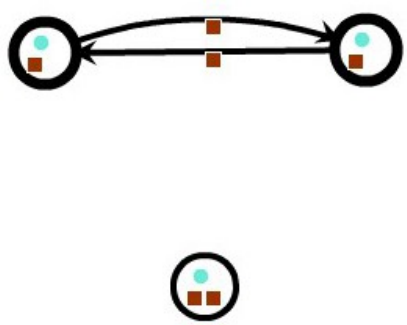}\end{minipage}\label{fig:abPG3}}\\
\subfigure[Since there were $a+1 = 2$ aqua pebbles on its endpoints, the red edge $g$ is inserted with an {\bf add red edge} move.]{\begin{minipage}[b]{.238\linewidth}\centering\includegraphics[width=\linewidth]
{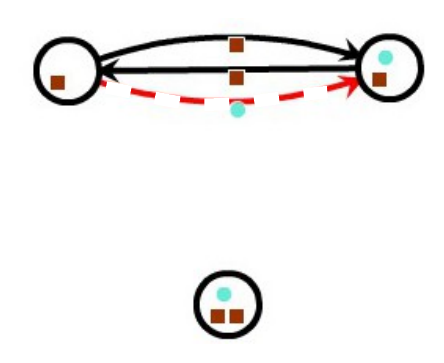}\end{minipage}\label{fig:abPG4}}
\hfil
\subfigure[A few moves later, an {\bf aqua edge exchange reversal} move is possible: $B$ has an aqua pebble and a tan pebble covered black in-edge $z$, whose endpoints are not in a (1,1)-component.]{\begin{minipage}[b]{.238\linewidth}\centering\includegraphics[width=\linewidth]
{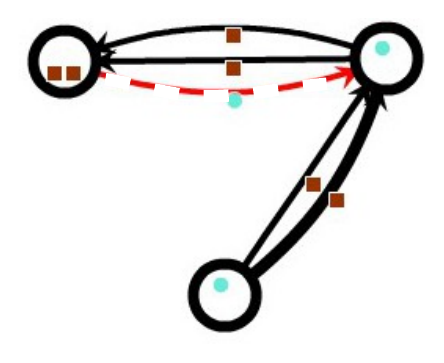}\end{minipage}\label{fig:abPG5}}
\hfil
\subfigure[The edge $z$ is subsequently reversed, covered by the aqua pebble from $B$, releasing a tan pebble onto $C$.]{\begin{minipage}[b]{.238\linewidth}\centering\includegraphics[width=\linewidth]
{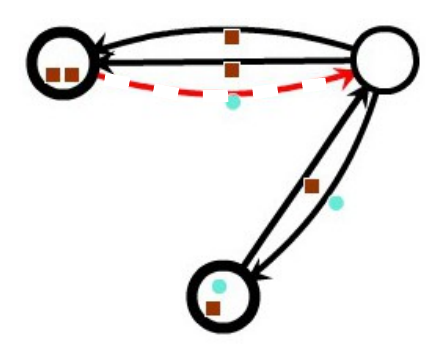}\end{minipage}\label{fig:abPG6}}
\hfil
\subfigure[Finally, all edges are successfully inserted with exactly $a+b=3$ pebbles remaining; the output is {\bf tight}.]{\begin{minipage}[b]{.238\linewidth}\centering\includegraphics[width=\linewidth]
{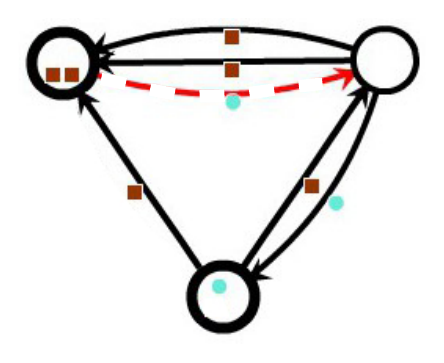}\end{minipage}\label{fig:abPG7}}
\caption{The $[a,b]$-pebble game (Algorithm \ref{alg:abPG}) played on a primitive cad graph for a generically minimally rigid framework determines that it is $[1,2]$-tight.}
\label{fig:pgExample}
\end{figure}

\subsection{Correctness}
We now show correctness of Algorithm \ref{alg:abPG}.
\begin{theorem}\theolab{pebblegame}
A bi-colored graph is $[a,b]$-sparse if and only if it can be constructed with the $[a,b]$-pebble game.
\end{theorem}

We are going to prove that \algref{abPG} correctly characterizes $[a,b]$-sparse graphs
and that it can be used to find circuits.  Structurally, the search procedure
in \algref{abPGfindPebble} corresponds to Knuth's \cite{Knuth:1973:MP:891978} algorithm
for matroid union (see \cite[Sec. 42.3]{S03b} for a modern treatment).  Thus, one can infer correctness
once this equivalence is established.  However, for clarity, we will first prove
that any graph constructed with the moves described in \algref{abPG} is $[a,b]$-sparse.

In what follows, we describe a pebble game configuration as $(H,A,T)$, where
$H$ is a bi-colored directed graph on vertex set $V$, with a pebble covering every
edge and some free pebbles on vertices; $A$ is the set of edges covered by
aqua pebbles and $T$ is the set of edges covered by tan pebbles. We will
also abuse notation slightly and use the same symbols $H$, $A$,
and $T$ to describe
their underlying undirected graphs. Finally, we will use the notation
$S + e$ to denote $S \cup \{e\}$ and $S - e$ to denote $S \setminus \{e\}$.

\begin{lemma}\lemlab{pgIsInd}
The underlying graph of any pebble game configuration is $[a,b]$-sparse
with $A$ as an $(a,a)$-sparse graph and $T$ as a $(b,b)$-sparse graph.
\end{lemma}
\begin{proof}
We show something slightly stronger, which is that the underlying graph $H$
constructed by applying \emph{any sequence} (as opposed to only the
ones found by the algorithm)
of the pebble game moves is always $[a,b]$-sparse.

The key invariant is that, after any sequence of moves, $A$ and $T$ both
induce pebble game configurations for the \emph{basic} (uncolored) pebble game from
\cite{leeStreinu}.
As an immediate
consequence, we obtain
that $A$ remains $(a,a)$-sparse and $B$ remains $(b,b)$-sparse.  This certifies that
$H$ is $[a,b]$-sparse.  The invariant clearly holds at initialization, so
we proceed by induction on the number of moves.

For the inductive step, we first consider
all the moves except for
the {\bf exchange edge reversal} moves.
We observe that,
assuming the required preconditions,
these operate entirely on either $A$ or $T$
as a configuration, so the inductive step for them follows directly
from \cite{leeStreinu}.

To complete the induction, consider the {\bf aqua edge exchange reversal}
move, since the tan one has an analogous proof.  The
precondition, that the edge $ij\in T$ is not in an $(a,a)$-component
of $A$, implies that $A + ij$ is $(a,a)$-sparse.  From this, the
pebble game invariants of \cite{leeStreinu} imply that there are
$a$ aqua pebbles (distinct from the one on $j$)
reachable from $i$ via paths using only edges in $A$.
By induction, $A$ could have been built by the basic
$(a,a)$-pebble game; then $ij$ could be added to $A$ by basic pebble
game searches. Notice that returning the tan-colored pebble to $j$ maintains
$T$ as a basic $(b,b)$-pebble game configuration.
\end{proof}

The other direction is captured by the following lemma.
\begin{lemma}\lemlab{pgInserts}
If an edge is independent of the underlying graph of a pebble game configuration,
then the pebble game will successfully insert it.
\end{lemma}

Before giving the proof, we briefly review Knuth's algorithm and establish
some terminology, specialized to our setup. The algorithm operates on a directed,
bipartite graph associated with a pebble game configuration $(H,A,T)$ and an edge $e$,
which is not in $H$.  This graph, denoted $\Gamma_{H+e}$, has vertex set given by
the edges of $H$, i.e., $A \disjointUnion T,$ along with two {\em terminal} vertices
$\alpha$ and $\tau$, and an additional vertex for the edge $e$.  First, we describe
the edges originating and terminating at a vertex $x\notin\{e,\alpha,\tau\}$.

There is a directed edge from vertex $x$ to $y$, written $x \to y$ if \LTcomment{Got rid of the footnote nobody likes.}
\[
\left\{
\begin{array}{ll}
(A - y) + x \text{ is $(a,a)$-sparse} & \quad \text{if } x \in T \text{ \& } y \in A\\
(T- y) + x \text{ is $(b,b)$-sparse} & \quad \text{if } y \in T \text{ \& }x \in A\cap B\text{, i.e., }\\ &x \text{ is a black edge in the aqua partition}
\end{array} \right.
\]
Additionally, there is an edge $x\to \alpha$ if $x\in T$ and $A + x$ is
$(a,a)$-sparse, and there is an edge from $x\to \tau$ if $x\in A\cap B$ and
$T + x$ is $(b,b)$-sparse.
The edges originating at $e$ are defined similarly. This case distinction
is simply to make it clear that no edges in $\Gamma_{H+e}$ have
$e$ as their target.

A path $x_0 \to x_1 \to \cdots \to x_\pi$ has a \emph{shortcut} in a graph
if there exists a $j > i+1$ such that $x_i \to x_j$ is an edge.
In particular, if $x_0 \to x_1 \to \cdots \to x_\pi$ is a shortest path
in a graph, it does not have a shortcut.

Given a path from $e$ to a terminal vertex, \emph{recoloring along
the path} means putting $x_i$ in the part of the partition
containing $x_{i+1}$, with $\alpha$ always in $A$ and $\tau$
always in $T$.

The main result of \cite{Knuth:1973:MP:891978}, again specialized for our setup, is:
\begin{theorem}\theolab{knuth}
Let $(H,A,T)$ be a pebble game configuration and $e$ an edge not in the underlying
graph.
Then there is a directed path in $\Gamma_{H+e}$ from $e$ to $\alpha$ or
$\tau$ if and only if $H+e$ is independent.  Moreover, given a
path $e = x_0 \to x_1 \to \ldots \to x_\pi\in\{\alpha,\tau\}$
in $\Gamma_{H+e}$ that does not have a shortcut, a partition
of $H+e$ certifying $[a,b]$-sparsity can be found by recoloring
along this path.
\end{theorem}
The proof of \lemref{pgInserts} amounts to showing that \algref{abPGfindPebble}
is simulating Knuth's algorithm.
\begin{proof}[Proof of \lemref{pgInserts}]
Assume that $e$ is independent of the underlying graph of a pebble game configuration
$H$.
We need to show that \algref{abPGfindPebble} will succeed in collecting
enough pebbles on the endpoints of $e$.  This is done by comparing
the pebble search procedure in \algref{abPGfindPebble}
to Knuth's algorithm.

First consider, in the main loop of \algref{abPGfindPebble},
the conditional block predicated upon when
$a+1$ (if $c$ is aqua) or $b+1$ (if $c$ is tan) pebbles can be collected on
the endpoints of $f$. Note that $a$ aqua or $b$ tan pebbles can always be
collected on any vertex by \cite{leeStreinu}.  The additional pebble can be collected
if and only if the edge $f$ can be moved to the opposite part of the partition without
violating sparsity.  This is equivalent to there being an edge
$f \to \{\alpha,\tau\}$ in $\Gamma_{H+e}$.

Otherwise, the pebble search fails.  In
this case, \cite{leeStreinu} implies that
$F' + f$ is the fundamental circuit of $f$ in the $c$-colored part of the partition;
i.e., $g\in F'$ if and only if there is an edge $f\to g$ in $\Gamma_{H+e}$.
Therefore, $F'$ is exactly the set of neighbors of $f$ in $\Gamma_{H+e}$.

By enqueuing those edges in $F'$ not already in $F$,
\algref{abPGfindPebble} is, in fact, searching
$\Gamma_{H+e}$ in a breadth-first fashion.
By \theoref{knuth}, the assumption that $H+e$ is $[a,b]$-sparse implies that
there is a path from $e$ to a terminal vertex in $\Gamma_{H+e}$.
Therefore, \algref{abPGfindPebble}
will be able to collect $a+1$ aqua or $b+1$ pebbles on the endpoints of some edge $f$,
implying that there is an edge from $f$ to a terminal in $\Gamma_{H+1}$.  Let $p$
be the path in $\Gamma_{H+e}$ defined by following predecessor pointers from
$f$.  Since \algref{abPGfindPebble} implements breadth-first search on $\Gamma_{H+e}$,
$p$ is shortcut free.

\theoref{knuth} then implies that recoloring along $p$ preserves the $(a,a)$-
and $(b,b)$-sparsity of $A$ and $T$ at every step.  The main results
of \cite{leeStreinu} then imply that it will always be possible to
meet the preconditions of the {\bf exchange edge reversal} moves
to implement the recoloring by using only basic pebble searches on $A$
or $T$.	 Thus, the pebble game moves implementing the recoloring along
$p$ will succeed, and the $[a,b]$-pebble game will insert $e$.
\end{proof}

\subsection{Circuits}
The pebble game also detects $[a,b]$-circuits, an approach that is perhaps less
well-known, but appears before in \cite[Section 6]{leeStreinu} and has been used
in \cite{BHMT11}.
Note that the presence of red edges create the possibility of many types of circuits.
Some may be circuits as uncolored $(a+b,a+b)$ graphs, and others may be $(a,a)$-circuits, but there
are other types.
The examples in \figref{circuits21} demonstrate a property of circuits that does not arise
in the $(k,\ell)$-sparsity
matroids. While every $(k,\ell)$-circuit is $(k,\ell)$-spanning, or ``rigid,''
an $[a,b]$-circuit may actually be ``flexible.''
Dropping an edge of a $(k,\ell)$-circuit always results in a tight graph, but dropping an edge of
an $[a,b]$-circuit can result in a sparse (but not tight) graph.

\begin{figure}[htb]
\subfigure[The fundamental {$[1,2]$}-circuit for rejected edge $g$ detected by the pebble game is the set of edges $\{e,g,y\}$. The set $\{e,y\}$ forms an angularly rigid block.]{\begin{minipage}[b]{.48\linewidth}\centering\includegraphics[scale=.4]
{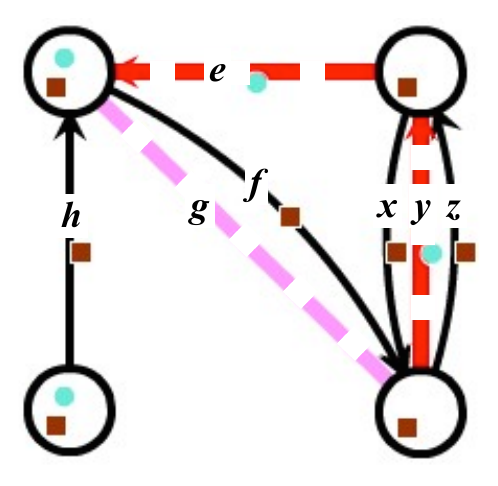}\end{minipage}\label{fig:abPG21circuit}}
\hfil
\subfigure[The fundamental {$[1,2]$}-circuit for rejected edge $y$ detected by the pebble game is the set of edges $\{e,g,x,y,z\}$. The set $\{x,z\}$ forms a contextually rigid block due to the presence of the angular block $\{e,g\}$.]{\begin{minipage}[b]{.48\linewidth}\centering\includegraphics[scale=.4]{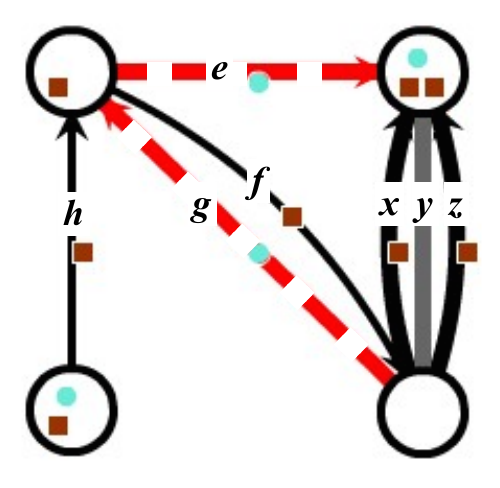}\end{minipage}\label{fig:circuit21}}
\caption{Unlike $(k,\ell)$-circuits, removing any edge from a $[1,2]$-circuit may not produce a $[1,2]$-graph.}
\label{fig:circuits21}
\end{figure}

Whenever we fail to insert an edge we can use \algref{abPGfindPebble}
to find its fundamental circuit.
\begin{lemma}Let $F$ be the set of edges returned
by \algref{abPGfindPebble}.  The fundamental circuit of $e$ in the configuration
graph $H$ is $F+e$.
\end{lemma}
\begin{proof}
We must show that $F+e$ is dependent and that, for any $y \in F$, $F + e - y$ is independent.
Observe that $F+e$ corresponds to the set of vertices reachable from $e$ in the Knuth graph $\Gamma_{H+e}$.
By the definition of $F$, every directed path in $\Gamma_{H+e}$ that starts at $e$ is
contained in $\Gamma_{F+e}$.  Therefore, there is no path from $e$ to a terminal in $\Gamma_{F+e}$, and
\theoref{knuth} implies that $F+e$ is not $[a,b]$-sparse.

Now, let $y \in F$. By construction of $F$, $y$ is on a short-cut free directed path starting at $e.$
Therefore,
there is an $x \in F+e$ on this path with $x \to y$ an edge of $\Gamma_{H+e}.$
By definition, removing
$y$ results in an edge from $x$ to a terminal, providing a path from $e$ to a terminal.
In other words, $F+e-y$ is $[a,b]$-sparse for all $y\in F$, which completes the
proof of correctness.
\end{proof}

\subsection{Complexity analysis}
The running time of  \algref{abPG} for a graph with $n$ vertices and $m$ edges
is $O(mn^2)$, which is $O(n^4)$.  First we observe that collecting the initial $a+b$ pebbles
in \algref{abPG} requires $O(n)$ for each edge (a total of $O(mn)$) and that the rest of the
steps may be charged to $O(m)$ invocations of \algref{abPGfindPebble}.

The running time of \algref{abPGfindPebble} is $O(n^2)$.  This
is because each of the $O(n)$ edges in the configuration is enqueued at most once in the main loop,
and each edge that is enqueued triggers a basic pebble game search
requiring $O(n)$ steps by \cite{leeStreinu}, after first copying a
configuration of size $O(n)$.

By way of comparison, a direct application of Knuth's algorithm leads to a more expensive running
time. In this approach, one might build
the bipartite graph explicitly and use the
basic pebble game to test each possible edge. The graph $\Gamma_{H+e}$ has $O(n^2)$ edges, and
each check would require an $O(n^2)$-time run of the basic pebble game; this would result in
a total of running time of $O(mn^4) = O(n^6)$.

\subsection{Finding components}
Within an $[a,b]$-sparse graph, an \emph{induced $[a,b]$-component} is a vertex-maximal
$[a,b]$-block. It is straightforward to adapt Algorithm \ref{alg:abPG} to maintain and
detect induced $[a,b]$-components, as in the $(k,\ell)$-pebble game algorithm with
components described in \cite{leeStreinu}. The running time would remain $O(n^4)$.

Note that any edge with vertices contained in an induced $[a,b]$-component is
dependent and will be rejected by this adapted pebble game in $O(1)$ time.
However, an edge may be dependent without being contained in an induced
$[a,b]$-component, as the circuits in \figref{circuits21} demonstrate. Therefore,
\emph{unlike the $(k,\ell)$-sparsity case}, we do not save a factor in the running time
by maintaining induced $[a,b]$-components.

\subsection{Detecting factor graphs}
We now describe algorithms for detecting factor graphs, useful in expressing
the pure condition. We begin by detecting
factor graphs in a $(k,k)$-graph before turning to $[a,b]$-graphs.

\subsubsection{The $k$-factor graphs algorithm.} Given a $(k,k)$-graph, Algorithm \ref{alg:kFactors} finds the factor graphs associated to the irreducible
factors of its pure condition. The basic intuition is that the algorithm detects minimal
rigid subgraphs, which are associated to irreducible factors, contracts them (into a
vertex representing a rigid body) and recurses.
\begin{algorithm}
\caption{The $k$-factor graphs algorithm.}
\label{alg:kFactors}
{\bf Input:} A $(k,k)$-graph $G = (V, E)$.\\
{\bf Output:} The factor graphs of  the irreducible factors of $C_G$.\\
{\bf Method:} \vspace{-3mm}
\begin{enumerate}
\item Initialize ${\cal P} = \emptyset$ and $H = G$.
\item Play the $(k,k+1)$-pebble game on $G$ to find a maximal $(k,k+1)$-sparse graph $G'$.
\item For every edge $e$ that is rejected (there is at least one):
\begin{enumerate}
\item Use the $(k,k+1)$-pebble game to detect the fundamental $(k,k+1)$-circuit $G_e$
of $e$ in $G'$.
\item Set ${\cal P} = {\cal P} +G_e$; set $H = H/G_e$.
\end{enumerate}
\item If $H$ is not a single vertex:
\begin{enumerate}
\item Recursively use the {\bf $k$-factor graphs algorithm}
on $H$ to obtain factor graphs ${\cal P'}$.
\item Set ${\cal P} = {\cal P} \cup {\cal P'}$.
\end{enumerate}
\item Output ${\cal P}$.
\end{enumerate}
\end{algorithm}
To show correctness, we need to check that every factor graph we find is \emph{irreducible} and \emph{proper}.
Properness comes from \theoref{kk-factor}.  Irreducibility of the pure condition of a $(k,k)$-graph
is characterized by White and Whiteley \cite[pg. 27]{whiteWhiteley}: the pure condition of a graph %
is irreducible
if and only if the graph contains no proper block. The graph contains no proper block if and only if,
for every proper
subgraph %
with $n'$ vertices and $m'$ edges, $m' < kn' - k$ (i.e., strict inequality holds on proper subgraphs).
Since we are considering a
$(k,k)$-graph
with all proper subgraphs $(k,k+1)$ sparse, this holds precisely when the graph is a $(k,k+1)$-circuit.

\subsubsection{The $[a,b]$-factor graphs algorithm.}
As shown in \secref{factors}, the factors of an $[a,b]$-graph have a more
complicated combinatorial structure than those of $(k,k)$-graphs. Therefore,
Algorithm \ref{alg:abFactors} adapts the $k$-factor graphs algorithm to rely on a
subroutine (Algorithm \ref{alg:abRedFactors}) that detects the additional types of factors,
both proper and improper.
These algorithms intuitively follow the same process of
detecting rigid components and contracting them; if a contextually rigid component is detected,
it is handled in Step \ref{alg:abRedFactorsContextual} of Algorithm \ref{alg:abRedFactors}.
\begin{algorithm}
\caption{The $[a,b]$-factor graphs algorithm.}
\label{alg:abFactors}
{\bf Input:} an $[a,b]$-graph $G = (V, E=R \disjointUnion B)$, with a set of red edges $R$ and
a set of black edges $B$.\\
{\bf Output:} Proper (irreducible) and improper factor graphs
of $G$ that together provide a factorization of $G$. \\
{\bf Method:} \vspace{-3mm}
\begin{enumerate}
\item Initialize ${\cal P} = \emptyset$; $\cal I = \emptyset$; $H = G$; set $k = a+b$.
\item Play the $(k,k+1)$-pebble game on $G$ to find a maximal $(k,k+1)$-sparse graph $G'$.
\item For every edge $e$ that is rejected (there is at least one):
\begin{enumerate}
\item Use the $(k,k+1)$-pebble game to detect the fundamental $(k,k+1)$-circuit $G_e$
of $e$ in $G'$.
\item Use the {\bf $[a,b]$-red-factor graphs algorithm} on $G_e$ to obtain sets of factor graphs
${\cal P}_e$ and ${\cal I}_e$.
\item Set $\cal P = \cal P \cup$$ {\cal P}_e$ and $\cal I =\cal I \cup$$ {\cal I}_e$;
set $H = H/G_e$.
\end{enumerate}
\item If $H$ is not a single vertex:
\begin{enumerate}
\item Recursively use the {\bf $[a,b]$-factor graphs algorithm}
on $H$ to obtain sets of factor graphs ${\cal P}'$ and ${\cal I}'$.
\item Set $\cal P = \cal P \cup $${\cal P}'$ and $\cal I = \cal I \cup $${\cal I}'$.
\end{enumerate}
\item Output ${\cal P}$ and $\cal I$.
\end{enumerate}
\end{algorithm}
\ASJcomment{do we union in I' or just set to I' in second to last step?}

\JScomment{Do we want $G$ to be an $[a,b]$-graph?}
\begin{algorithm}
\caption{The $[a,b]$-red-factor graphs algorithm.}
\label{alg:abRedFactors}
{\bf Input:} An $[a,b]$-bi-colored graph $G = (V, E=R \disjointUnion B)$, with a set of red edges $R$ and
a set of black edges $B$, that contains no proper $(a+b,a+b)$-components.\\
{\bf Output:} Proper (irreducible) and improper factor graphs
of $G$ that together provide a factorization of $G$. \\
{\bf Method:} \vspace{-3mm}
\begin{enumerate}
\item Play the $(a,a)$-pebble game on $(V, R)$ and detect red $(a,a)$-components.
\item If no $(a,a)$-components of $(V,R)$ were found, output $\{G\}$.
\item Otherwise, let $H_1,\ldots, H_t$ be the red $(a,a)$-components.
\begin{enumerate}
\item For each graph $(V(H_i),E(V(H_i))\cap B)$, play the $(b,b)$-pebble game
to detect $(b,b)$-components.
\item If no $(b,b)$-components are found, use the {\bf $a$-factor graphs algorithm} on
each $H_i$ to obtain a set of proper factor graphs $\cal P$ that are the
$a$-factor graphs of the $H_i$.  Then
return $\cal P$ and ${\cal I} = \{(V, E \setminus \{E(V(H_i)) : i\in [t]\} )\}$.
\item \label{alg:abRedFactorsContextual} Otherwise, there is a $(b,b)$-component $J$ in the vertex
span of some component $H_i$. Set $H = H_i$.
\begin{enumerate}
\item Use the {\bf $a$-factor graphs algorithm} on $H$ to obtain a set of
factor graphs ${\cal P}_H$.
\item Use the {\bf $b$-factor graphs algorithm} on $J$ to obtain a set of
factor graphs ${\cal P}_J$.
\item Use the $[a,b]$-pebble game to find an $[a,b]$-fan of $G$ and collect
the remaining $k$ pebbles on a vertex %
in $V(J)$.
Delete every edge of $H$ with its
tail in $V(J)$.  Contract $V(J)$ into a single vertex to
get a graph $G'$.
\item Use the {\bf $[a,b]$-factor graphs algorithm} on $G'$ to get
a set of factor graphs ${\cal P}'$ and ${\cal I}'$.
\item Find the factor graph $F_{H}$ in ${\cal P}'$ that contains an edge of $H$.
\item Return ${\cal P}_{H} \cup {\cal P}_J \cup ({\cal P}'\setminus \{F_{H}\})$
and ${\cal I}'$.
\end{enumerate}
\end{enumerate}

\end{enumerate}
\end{algorithm}

We provide an overview of how
the $[a,b]$-factor graphs algorithm performs on a $[1,2]$-graph with only proper factors
in \figref{example21recurse}
and on a $[1,1]$-graph with one improper factor in \figref{example11improper}.
A more comprehensive trace of the algorithm is given in on a $[1,1]$-graph
in \figref{factorAlgEx11}.

\begin{figure}[hbt]
\centering\includegraphics[scale=.5]{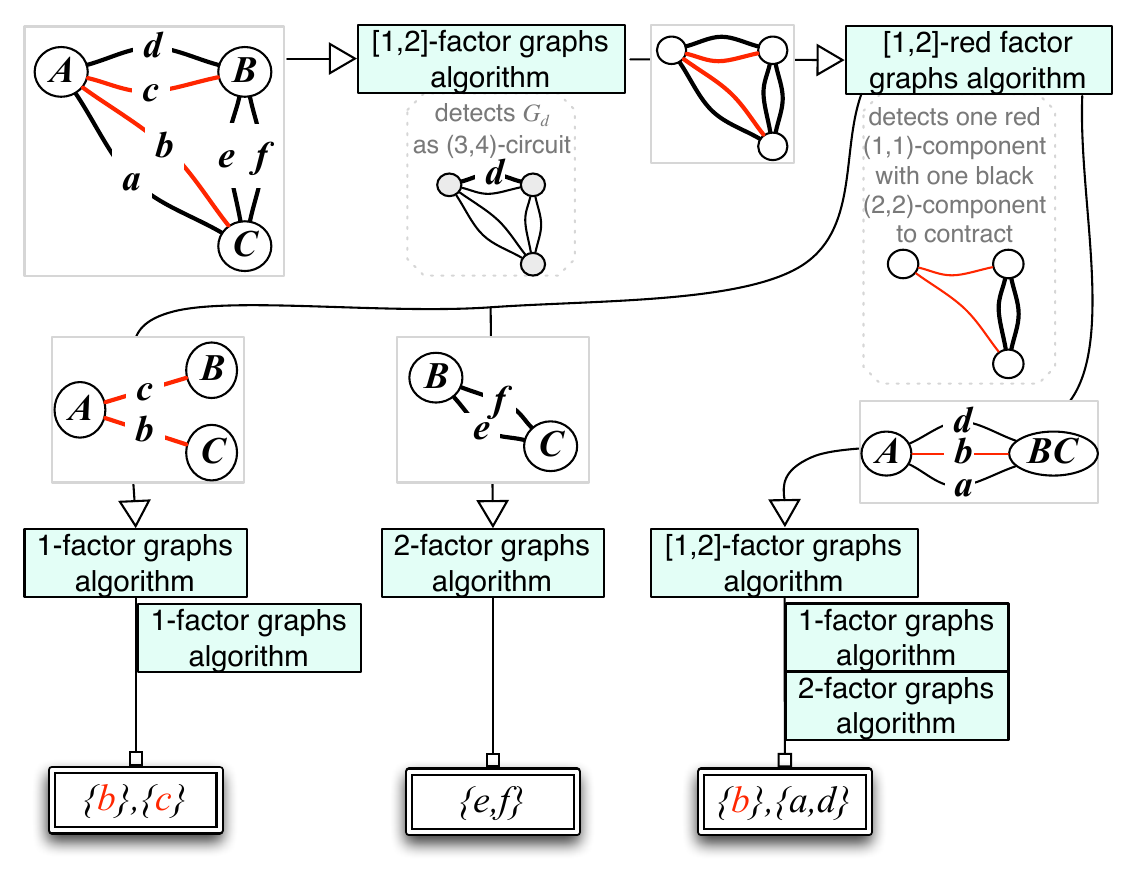}
\caption{Algorithm \ref{alg:abFactors} finds the proper factor graphs of the {$[1,2]$}-graph underlying \figref{gen3Body2D}: $\{\{e,f\},\{a,d\},\{b\},\{c\}\}$; there are no improper factor graphs. The pure condition is {$\pm [ef]_{(1,2)}[ad]_{(1,2)}[b]_{(3)}[c]_{(3)}$}. A bracket subscripted by ordered tuple $T$ denotes the determinant of the $|T| \times |T|$ matrix with coordinates specified by $T$.}
\label{fig:example21recurse}
\end{figure}

\begin{figure}[hbt]
\centering\includegraphics[scale=.5]{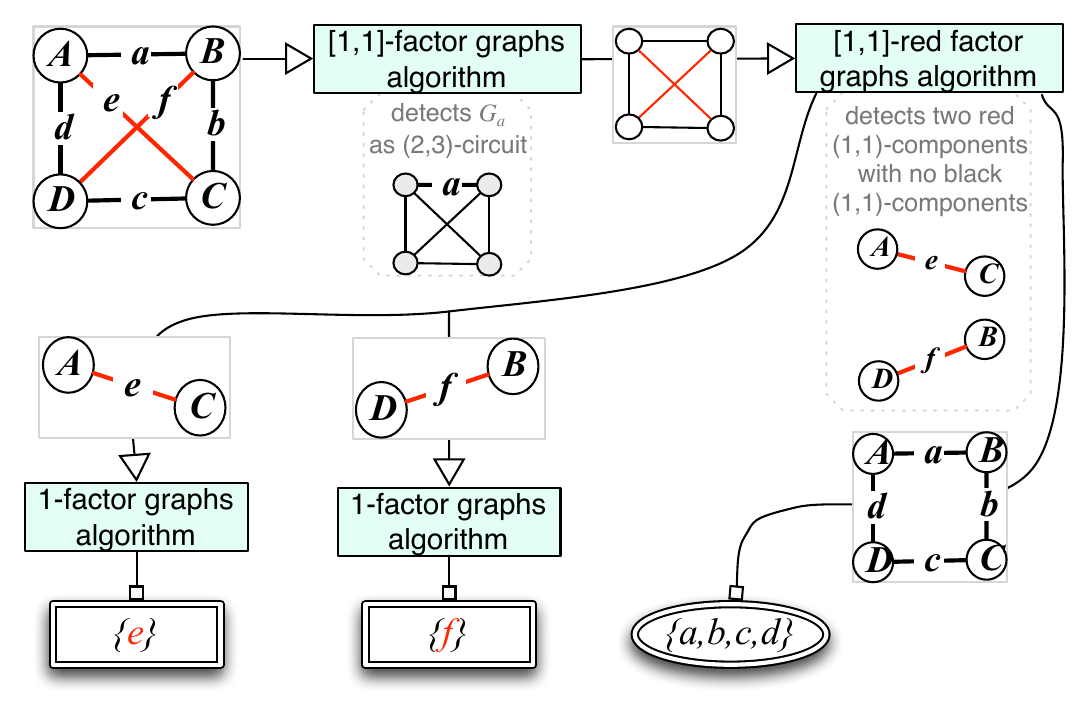}
\caption{Algorithm \ref{alg:abFactors} on a {$[1,1]$}-graph finds proper factor graphs $\{\{e\},\{f\}\}$ and an improper factor graph $\{\{a,b,c,d\}\}$. The pure condition is {$[e]_{(1)}[f]_{(1)}$}$(a_2b_1c_1d_1 - a_1b_2c_1d_1 - a_1b_1c_2d_1 + a_1b_1c_1d_2)$ = $e_1f_1(a_2b_1c_1d_1 - a_1b_2c_1d_1 - a_1b_1c_2d_1 + a_1b_1c_1d_2)$. A bracket subscripted by ordered tuple $T$ denotes the determinant of the $|T| \times |T|$ matrix with coordinates specified by $T$.
}
\label{fig:example11improper}
\end{figure}

We first prove correctness of Algorithm \ref{alg:abRedFactors}.
\begin{claim}
Algorithm \ref{alg:abRedFactors} returns only factor graphs; the proper factor
graphs are irreducible.
\end{claim}
\begin{proof}

The calls to the $a$-factor graphs algorithm and $b$-factor graphs algorithm produce irreducible
proper factor graphs by the correctness of Algorithm \ref{alg:kFactors}.  In Step 3.(c)iii we have an $[a,b]$-fan
obtained from a pebble
game configuration certifying $(a,a)$-sparsity of the aqua partition and $(b,b)$-sparsity of the
tan partition.  Via sparsity, we also know that there exists an $[a,b]$-tree decomposition that
gives rise to this $[a,b]$-fan if we direct all edges towards the tie-down.  Therefore, by
\theoref{bb-factor} we know that $G'$ is an $[a,b]$-graph.

Also by \theoref{bb-factor}, $H' = H / J$ is a (red) $(a,a)$-block.
Thus, after contracting and recursively calling Algorithm \ref{alg:abFactors},
there is exactly
one factor graph $H'$ containing an edge from $H$, so
Step 3(c)v is well-defined. Finally, removing the factor graph $F_H = H'$ from the set of
returned factor graphs completes the algorithm's correctness.
\end{proof}

This, along with  \theoref{kk-factor},
allows us to conclude correctness of our main Algorithm \ref{alg:abFactors}:
\begin{claim}
Algorithm \ref{alg:abFactors} returns only factor graphs, and the proper factor
graphs are irreducible.
\end{claim}
We do not know if the improper factor graphs are irreducible, since we do not
have a nice representation for them.

\begin{ex}[name={Example \ref{ex:22}, continued for the last time}]
In the $(k,k)$-setting, if an irreducible factor of $P_G$ is zero,
we may identify a $(k,k)$-subgraph of $G$ that is in a non-generic position, and there is a
dependence relation on the rows of  the submatrix corresponding to this subgraph.

In our example, setting each factor equal to zero implies
the existence of a dependence among the rows of the rigidity matrix of $G$, but there is an important distinction
between the supports of these dependence relations.  If $(a_1b_2-a_2b_1)=0$ we are guaranteed to
get a dependence supported on the first two rows of the rigidity matrix (which contain $a$ and $b$).
However, if $(c_3d_4-c_4d_3)=0$, the associated linear dependence \emph{must}  involve all
four rows of the rigidity matrix.  This is because there are no conditions on $c_1, c_2, d_1,$ and $d_2$
implied by $(c_3d_4-c_4d_3)=0$.\LTcomment{Check this.  Rewording.}

What this example is showing is that the rigidity matrix may drop rank because
a polynomial supported on a subset of edges vanishes, yet the corresponding dependence in the
rigidity matrix may be supported on a set of rows indexed by a larger subset of edges.
This phenomenon complicates the correspondence between the combinatorics of
$G$ as an $[a,b]$-graph and factors of $P_G.$ \LTcomment{Check this. Rewording.}

\end{ex}

\begin{question}
\label{conj:irred}
Are all the factor graphs found by Algorithm \ref{alg:abFactors} are irreducible?
\end{question}

\section{A case study}
\seclab{caseStudy}

In this section we show how a geometric interpretation of the vanishing of the pure condition can be used to predict special positions of the 2D body-and-cad framework consisting of 3 bodies, 2 bars, and 2 line-line coincidence constraints depicted in \figref{gen3Body2D}.

The associated primitive cad graph, in which an edge corresponds to a linear constraint, is given in \figref{threeBodyPrimCad}.
\begin{figure}[tb]
\subfigure[The primitive cad graph for the framework in \figref{gen3Body2D}.  Red edges represent angular constraints.]{\begin{minipage}[b]{.25\linewidth}\centering\includegraphics[width=\linewidth]
{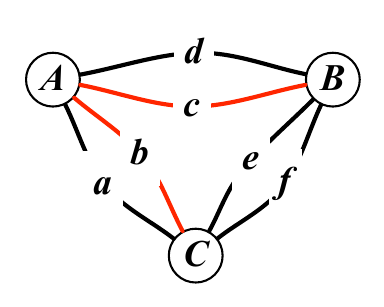}\end{minipage}\label{fig:threeBodyPrimCad}}
\hfil
\subfigure[The special position in which $\ell_1 \| \ell_2$.]{\begin{minipage}[b]{.32\linewidth}\centering\includegraphics[scale=.15]
{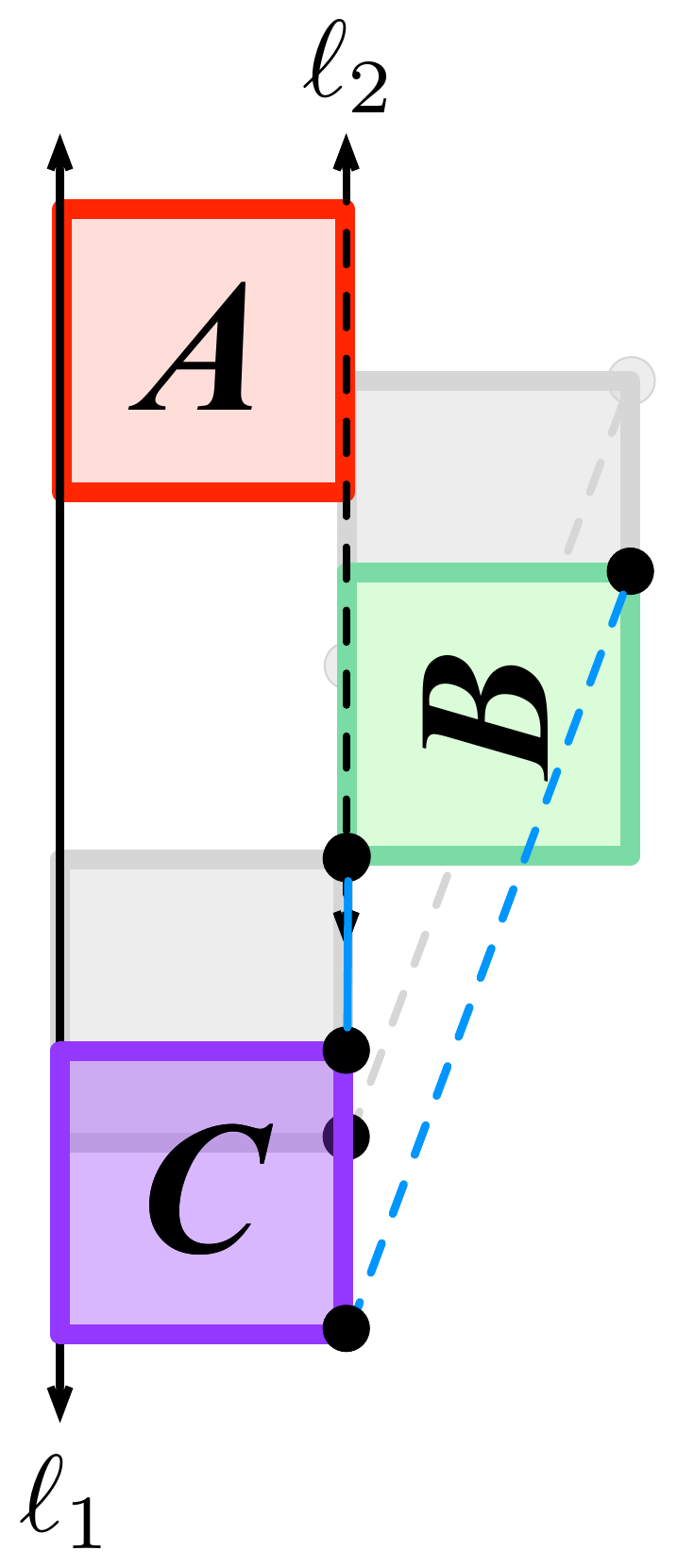}\end{minipage}\label{fig:threeBodySpecial1}}
\hfil
\subfigure[The special position in which the bars are parallel.]{\begin{minipage}[b]{.32\linewidth}\centering\includegraphics[scale=.15]
{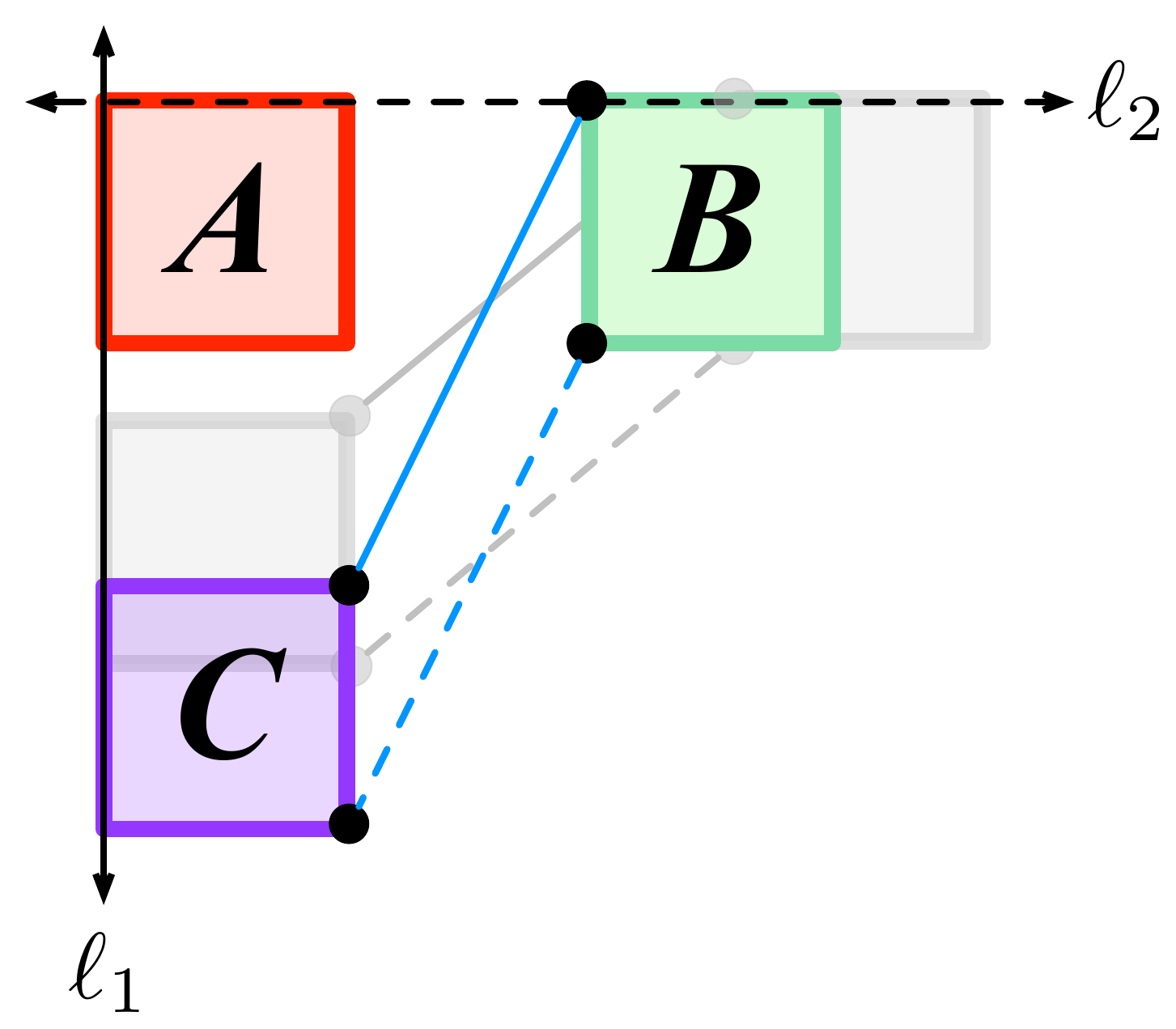}\end{minipage}\label{fig:threeBodySpecial2}}
\caption{A 2D body-and-cad example on 3 bodies; we assume that body $C$ is tied down.}
\label{fig:specialPositions3body}
\end{figure}
In this graph, each line-line coincidence is represented by a red edge, corresponding to a line-line parallel constraint (which restricts only angular motion), and a black edge corresponding to a point-line coincidence constraint (which restricts one translational degree of freedom).   Each bar eliminates 1 degree of freedom and is represented by a black edge.  %

We will realize this framework in the projective plane $\PP^2$, which
allows us to unify the treatment of rotations and translations if we view an infinitesimal translation as a rotation about a point on the line at infinity.  (See \cite{whiteWhiteley} and \cite{halleretal} for a detailed introduction to this point of view.)
The framework will lie in the affine piece of $\PP^2$ with coordinates $[x:y:1]$.  From this point of view, the line at infinity is $[x:y:0]$ and parallel lines in the $[x:y:1]$-plane meet at a point on the line at infinity corresponding to their slope.

A 2D infinitesimal rigid motion can be represented by a point in $(\PP^2)^*$ which is dual to the space in which we embed the framework, and a rotation is represented by a point whose first two coordinates are zero. %
Note that the point $[0:0:1] \in (\PP^2)^*$ is dual to the line at infinity in our original $\PP^2.$

If we tie down body $C$ and consider the underlying $(3,3)$-graph, the pure condition is the bracket polynomial $[abc][def]-[abd][cef].$  Since the edges labeled $b$ and $c$ correspond to angular constraints, and their
\ASJcomment{should be ``first two coordinates are zero'' but I don't know if that means other things have to change?}
last two coordinates are zero, the bracket $[abc]$ evaluates to zero.  Therefore, in this case the pure condition is a bracket monomial, $[abd][cef],$ which vanishes when either factor vanishes.
Note that the brackets $[abd]$ and $[cef]$ represent reducible polynomials in the coordinates of $a,\ldots, f.$  In fact,
$[abd] = b_3(a_1d_2-a_2d_1),$ and $[cef] = c_3(e_1f_2-e_2f_1).$
Therefore, $C_G$ has four irreducible factors. Compare to the output of Algorithm \ref{alg:abFactors} in \figref{example21recurse}.

The vanishing of each bracket factor has a geometric interpretation that is obscured in the full polynomial form.
The bracket $[abd]$ vanishes when the three points $a,b,$ and $d$ are collinear in $(\PP^2)^*.$  This means that the lines dual to these three points meet at a point in $\PP^2.$  Since the line dual to $b$ is the line at infinity, this means that the lines dual to $a$ and $d$ meet at a point on the line at infinity.  Consequently, these two lines are parallel.  In terms of our original geometric constraints, this shows that if lines $\ell_1$ and $\ell_2$ are parallel, then the body-and-cad framework is in a special position, and the framework admits an internal motion as depicted in Figure \ref{fig:threeBodySpecial1}.  Similarly, the bracket $[cef]$ vanishes when the lines dual to $e$ and $f$ are parallel in $\PP^2,$ which is shown in Figure \ref{fig:threeBodySpecial2}.
Thus, this analysis of the factors of the pure condition associated to the design in \figref{swSketchSpecial}
leads to the useful feedback that ``these two lines being parallel cause a dependency,'' prompting the problem of automating
such an analysis generally, which will require factoring in the \emph{Grassmann-Cayley algebra} when the pure condition is not just a product of brackets.

\section{Conclusions}
\seclab{OpenQuestions}
The approach presented in this paper is part of a larger research path to
provide computational tools that will
give users information about dependencies present in CAD structures
in terms of the original geometric constraint system.
A prototype of \algref{abPG} has been implemented, with a long-term goal
to see the pebble game and factor algorithms incorporated into
commercial CAD software packages.
By analyzing the pure condition, we can detect special positions of a generically minimally
rigid body-and-cad structure.
However, since $C_G$ vanishes when $G(\bp)$ is {\em infinitesimally} flexible,
special positions that we find may not be truly flexible.
These positions may still be of interest to a CAD user, as an
infinitesimally flexible framework carries an internal stress,
indicative of structural weaknesses.
Moreover, we may be able to combine conditions implying a special position to create
degenerate embeddings with true motions.
We conclude with
a brief discussion of open questions that arise as we move toward further development of our approach.

Algorithm \ref{alg:abFactors} returns factor graphs of the pure condition of an $[a,b]$-graph,
but it remains open as to whether these factors are irreducible or not.%
When $b=0,$ the results of White and Whiteley \cite{whiteWhiteley} show that irreducible factors of $C_G$ correspond to circuits in $G$; this correspondence implies that Algorithm \ref{alg:kFactors} produces irreducible factors for body-and-bar graphs.
A better understanding of circuits and stresses would allow us to similarly conclude that
the factors identified by Algorithm \ref{alg:abFactors}.

We were able to carry out an analysis in the case study of \secref{caseStudy}
where the pure condition was just a product of brackets, and its vanishing was implied by either making two bars
parallel or two lines parallel.
More generally, a geometric interpretation of the vanishing of a more complicated non-monomial bracket
polynomial may be possible via the process of Cayley factorization, which takes as input
a polynomial written in terms of brackets of vectors and outputs an expression in terms of meets and joins in the Grasmann-Cayley algebra
of those points if such an expression exists.  There is a Cayley factorization algorithm due to White \cite{whiteCayley, sturmfels}, and it
would be interesting to see if it could be modified (and sped up) if the input bracket polynomial is known to be a pure
condition.

Even when a Cayley factorization does exist, it may be nontrivial to extract geometric information
about the original framework from it.  One issue that adds complexity in general is that a single
\emph{cad} constraint
may impose multiple linear constraints, so conditions may need to be expressed in terms of sets of vectors.
Furthermore, in 3D, the vectors
in the brackets do not live in a space dual to our realization space (as they do in 2D), complicating
translation of the vanishing of the pure condition into the setting of our original constraints.

Finally, the results in this work rely on the combinatorial characterization of \cite{stjohnSidman}, which
apply to 3D body-and-cad structures \emph{without point-point coincidence constraints}. While a combinatorial
characterization that incorporates these constraints remains unknown,
3D body-and-cad frameworks with point-point coincidences share similar properties with presumed
barriers to a combinatorial characterization of 3D bar-and-joint frameworks.

\begin{figure}[bt]
\centering\includegraphics[width=.9\linewidth]{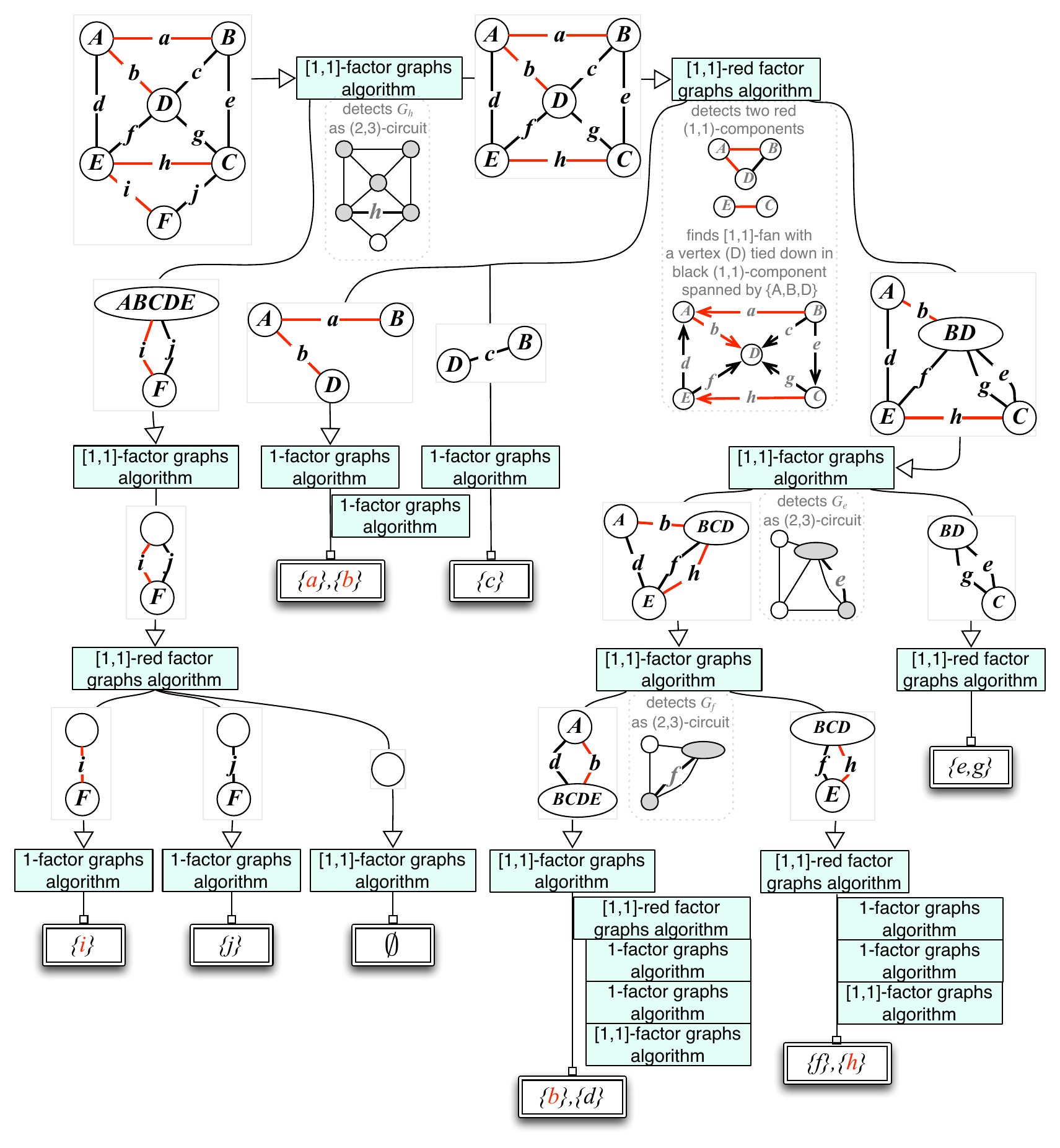}
\caption{Algorithm \ref{alg:abFactors} on a $[1,1]$-graph detects proper factor graphs: $\{\{a\},\{b\},\{c\},\{i\},\{j\},\{f\},\{h\},\{d\}\{e,g\}\}$. There are no improper factor graphs. The pure condition is $\pm [a]_{(1)}[b]_{(1)}[c]_{(2)}[i]_{(1)}[j]_{(2)}[f]_{(2)}[h]_{1)}[d]_{(2)}[eg]_{(1,2)}$.}
\label{fig:factorAlgEx11}
\end{figure}
\section*{\bf Figures and acknowledgements. }
Some figures created in SolidWorks 2010 and 2012.
We would like to thank Ruimin Cai for implementing some of the algorithms presented here and
Josephine Yu for her work in implementing the Cayley factorization algorithm
at the Macaulay2 workshop (Colorado College, Aug. 8-12, 2010; supported by
NSF Grant No. 0964128 and NSA Grant No. H98230-10-1-0218.)
We are grateful for the thoughtful and constructive feedback from anonymous reviewers.

\bibliographystyle{elsarticle-num}

\end{document}